\documentclass{article}
\usepackage[utf8]{inputenc}
\usepackage{amsmath}
\usepackage{amsfonts}
\usepackage{amssymb}
\usepackage{amsthm}
\usepackage{subcaption}
\theoremstyle{definition}
\newtheorem{theorem}{Theorem}

\newtheorem{definition}{Definition}

\newtheorem{corollary}{Corollary}

\newtheorem{claim}{Claim}[section]

\usepackage{textcomp}

\usepackage[utf8]{inputenc}

\usepackage[version=3]{mhchem}
\usepackage{commands-tam}
\usepackage{graphicx}
\usepackage{url}
\usepackage{wrapfig}
\usepackage{cite}
\usepackage{marginnote}				%
\usepackage[textsize=scriptsize]{todonotes} 

\usepackage[normalem]{ulem} %

\makeatletter
\newcommand{\oset}[3][0ex]{%
  \mathrel{\mathop{#3}\limits^{
    \vbox to#1{\kern-2\ex@
    \hbox{$\scriptstyle#2$}\vss}}}}
\makeatother

\renewcommand{\vec}[1]{{\ensuremath{\mathbf{#1}}}}
\newcommand{\vecofvec}[1]{{\ensuremath{\oset[-0.2ex]{\rightharpoonup}{\mathbf{#1}}}}}

\newcommand{\vvc}{\vecofvec{c}}

\newcommand{\vm}{\vec{m}}

\renewcommand{\vb}{\vec{b}}
\newcommand{\vc}{\vec{c}}

\usepackage[font={footnotesize}]{caption}

\newcommand{\calM}{\mathcal{M}}

\let\OLDthebibliography\thebibliography
\renewcommand\thebibliography[1]{
  \OLDthebibliography{#1}
  \setlength{\parskip}{0pt}
  \setlength{\itemsep}{0pt plus 0.3ex}
}

\title{Thermodynamically Favorable Computation via Tile Self-assembly}

\author{
 Cameron Chalk
    \thanks{Department of Electrical and Computer Engineering, University of Texas at Austin, Austin, TX, USA \protect\url{ctchalk@utexas.edu} This author's research was supported in part by National Science Foundation Grants CCF-1618895 and CCF-1652824.}
\and
 Jacob Hendricks%
    \thanks{Department of Computer Science and Information Systems, University of Wisconsin - River Falls, River Falls, WI, USA
    \protect\url{jacob.hendricks@uwrf.edu}}
\and
 Matthew J. Patitz
    \thanks{Department of Computer Science and Computer Engineering, University of Arkansas, Fayetteville, AR, USA
    \protect\url{patitz@uark.edu} This author's research was supported in part by National Science Foundation Grants CCF-1422152 and CAREER-1553166.}
\and
    Michael Sharp
        \thanks{Department of Computer Science and Computer Engineering, University of Arkansas, Fayetteville, AR, USA.
        \protect\url{mrs018@uark.edu}.
        This author's research was supported in part by National Science Foundation Grants CCF-1422152 and CAREER-1553166.}
}

\date{}

\begin{document}

\maketitle

\begin{abstract}
    The recently introduced Thermodynamic Binding Networks (TBN) model was developed with the purpose of studying self-assembling systems by focusing on their thermodynamically favorable final states, and ignoring the kinetic pathways through which they evolve.  The model was intentionally developed to abstract away not only the notion of time, but also the constraints of geometry.  Collections of monomers with binding domains which allow them to form polymers via complementary bonds are analyzed to determine their final, stable configurations, which are those which maximize the number of bonds formed (i.e. enthalpy) and the number of independent components (i.e. entropy). In this paper, we first develop a definition of what it means for a TBN to perform a computation, and then present a set of constructions which are capable of performing computations by simulating the behaviors of space-bounded Turing machines and boolean circuits. In contrast to previous TBN results, these constructions are robust to great variability in the counts of monomers existing in the systems and the numbers of polymers that form in parallel.  Although the Turing machine simulating TBNs are inefficient in terms of the numbers of unique monomer types required, as compared to algorithmic self-assembling systems in the abstract Tile Assembly Model (aTAM), we then show that a general strategy of porting those aTAM system designs to TBNs produces TBNs which incorrectly simulate computations.  Finally, we present a refinement of the TBN model which we call the Geometric Thermodynamic Binding Networks (GTBN) model in which monomers are defined with rigid geometries and form rigid bonds.  Utilizing the constraints imposed by geometry, we then provide a GTBN construction capable of simulating Turing machines as efficiently as in the aTAM.
\end{abstract}

\section{Introduction}

The study of self-assembling systems has resulted in a wide range of theoretical models and results, showing powers and limitations of such systems across a large landscape of variation in component structures, dynamics, and other important system properties \cite{PatitzSurvey,DotCACM,Versus,Polygons,Polyominoes,jSignals,FlexibleVsRigid,AGKS05g,WoodsMeunierSTOC}.  Theoretical studies have also given rise to experimental implementations in which artificial self-assembling systems are being developed and demonstrated in laboratories \cite{evans2014crystals,SchulmanWinfreeNucleation,SchulmanWinfreeNucleation,SchWin07}. While at times theoretical studies are intended solely to explore the mathematical boundaries between the possible and impossible%
, at other times they are geared toward informing researchers on the behaviors of existing physical systems.  Toward that end, theoretical models of self-assembly have been developed which seek to elucidate errors observed in experimental implementations (e.g. the kinetic Tile Assembly Model\cite{Winf98}), and among these is the Thermodynamic Binding Network (TBN) model\cite{TBNintro}.  Although such models are generally intended to abstract away many of the details of physical systems, they are often designed to highlight certain important aspects and isolate them for study.  For the TBN model in particular, the desire is to focus on thermodynamically favored end states of systems while ignoring the kinetic pathways through which they evolve, in the hopes of being able to better design systems whose ``sink states'' will be those we desire, \emph{and avoid those we don't}, regardless of intermediate states which may be traversed along the way.  Preliminary work with the TBN model \cite{TBNintro,arxivTBNstable} has provided initial tools to begin working with self-assembling systems in this model, and the goal of this paper is to extend them to TBN systems capable of performing a larger class of computations.  Since the notion of what it means to compute in such a model is not obvious, we provide a definition of computing with TBNs.  We then present a result showing that for any space-bounded Turing machine there exists a TBN which can simulate it on any input.  However, the size of the set of monomers required for the construction is on the order of the amount of space used multiplied by the number of time steps of the machine, making it less efficient than typical algorithmic self-assembling systems.  Nonetheless, the construction is robust to the system containing multiple copies of the computation simulation self-assembling in parallel, and also to large ranges of the numbers of monomers of each type, which is in contrast to the previous results.
We further extend our construction to the simulation of arbitrary fan-in fan-out Boolean circuits.

Next, we present results which relate the simulation of computations within the abstract Tile Assembly Model (aTAM)\cite{Winf98} to TBNs, as there are many results related to computation in the aTAM (e.g. \cite{jSADS,jCCSA,CookFuSch11,SolWin07,OptimalShapes3D,WoodsIU2013,jLSAT1}), which has been shown to be computationally universal \cite{Winf98}. Along this line, we first consider a standard class of aTAM systems which are used to simulate Turing machines (i.e. ``zig-zag'' systems) and consider what happens if the tiles of those systems are interpreted as monomers of a TBN (in a straightforward manner).  We present a set of criteria which are relatively natural and likely to be met by the computations performed by many Turing machines and which, if true for a particular Turing machine, demonstrate why a TBN created in such a way would not correctly simulate the Turing machine's computations, and would be capable of outputting incorrect answers. Notably, the argument presented pertains to all of our currently known approaches to simulating computations using TBNs which are more monomer-efficient than those of our first result, leaving an open question of whether or not more efficient simulation is possible in the TBN model.  For our final result, we present a refinement to the TBN model which we call the Geometric Thermodynamic Binding Networks (GTBN) model, in which monomers and bonds between them are restricted by geometric constraints (unlike in the TBN model, but similar to the aTAM), and show how GTBNs can efficiently simulate arbitrary Turing machines for decidable languages.

\section{Preliminaries}

In this section we provide definitions for the TBN model.  Due to space constraints, definitions for the abstract Tile Assembly Model and zig-zag assembly systems can be found in Sections~\ref{sec:tam-informal} and \ref{sec:zig-zag}, respectively.

\subsection{TBN Model}\label{sec:model}

We use the definitions from \cite{TBNintro}, the majority of which we repeat here, but please see \cite{TBNintro} for more details and examples.%

Let $\N,\Z,\Z^+$ denote the set of nonnegative integers, integers, and positive integers, respectively.
A key type of object in our definitions is a multiset, which we define in a few different ways as convenient.%
Let $\calA$ be a finite set.
We can define a multiset over $\calA$ using the standard set notion, e.g., $\vc = \{a, a, c\}$, where $a, c \in \calA$.
Formally, we view multiset $\vc$ as a vector assigning counts to $\calA$.
Letting $\N^\calA$ denote the set of functions $f:\calA\to\N$,
we have $\vc \in \N^\calA$. 
We index entries by elements of $a \in \calA$, calling $\vc(a) \in \N$ the \emph{count of $a$ in $\vc$}.

Molecular bonds with precise binding specificity are modeled abstractly as binding ``domains'',
designed to bind only to other specific binding domains.
Formally, consider a finite set $\calD$ of \emph{primary domain types}.
Each primary domain type $a \in \calD$ is mapped to a \emph{complementary domain type} (a.k.a., \emph{codomain type}) denoted $a^*$.
Let $\calD^* = \{a^* \mid a \in \calD\}$ denote the set of codomain types of $\calD$.
The mapping is assumed 1-1, so $|\calD^*| = |\calD|$.
We assume that a domain of primary type $a \in \calD$ binds only to its corresponding complementary type $a^* \in \calD^*$, and vice versa.%
The set $\calD \cup \calD^*$ is the set of \emph{domain types}.

We assume a finite set $\calM$ of \emph{monomer types}, where a monomer type $\vm \in \N^{\calD \cup \calD^*}$ is a non-empty multiset of domain types, e.g., $\vm = \{a, b, b, c^*, a^*\}$. %
A  \emph{thermodynamic binding network} (TBN) is a pair $\calT = (\calD,\calM)$ consisting of a finite set $\calD$ of primary domain types and a finite set $\calM \subset \N^{\calD \cup \calD^*}$ of monomer types.
A \emph{monomer collection} $\vvc \in \N^\calM$ of $\calT$ is multiset of monomer types;
intuitively, $\vvc$ indicates how many of each monomer type from $\calM$ there are, but not how they are bound.%
Since one monomer collection usually contains more than one copy of the same domain type,
we use the term \emph{domain} to refer to each copy separately.%
We similarly reserve the term \emph{monomer} to refer to a particular instance of a monomer type if a monomer collection has multiple copies of the same monomer type.

A single monomer collection $\vvc$ can take on different configurations depending on how domains in monomers are bound to each other.
To formally model configurations, we first need the notion of a bond assignment.
 Let $(U,V,E)$ be the bipartite graph  describing all possible bonds, where
$U$ is the multiset of all primary domains in all monomers in $\vvc$,
$V$ is the multiset of all codomains in all monomers in $\vvc$,
and $E$ is the set of edges between primary domains and their complementary codomains $\{ \{u,v\} \mid u \in U, v \in V, v = u^*\}$.
A \emph{bond assignment} $M$ is  a matching %
on $(U,V,E)$.
Then, a \emph{configuration} $\alpha$ of monomer collection~$\vvc$ is  the (multi)graph 
$(U \cup V,E_M)$, where 
the edges $E_M$ include both the edges in the matching $M$ and an edge between each pair of domains within the same monomer.%
Specifically, for each pair of domains $d_i,d_j \in \calD \cup \calD^*$ that are part of the same monomer in $\vvc$, let $\{d_i,d_j\} \in E_M$, calling this a \emph{monomer edge}, and for each edge $\{d_i,d_i^*\}$ in the bond assignment $M$, let $\{d_i,d_i^*\} \in E_M$, calling this a \emph{binding edge}.
Let $[\vvc]$ be the set of all configurations of a monomer collection $\vvc$.
For a configuration $\alpha$, we say the size of a configuration, written $|\alpha|$, is simply the number of monomers in it.
Each connected component in $\alpha$ is called a \emph{polymer}. %
Note that a polymer is itself a configuration, but of a smaller monomer collection $\vvc' \subseteq \vvc$ (as $\vvc'$ and $\vvc$ are  multisets). %
As with all configurations, the size of a polymer is the number of monomers in~it.

Which configurations are thermodynamically favored over others depends on two properties of a configuration: its bond count and entropy.
The \emph{enthalpy} $H(\alpha)$ of a configuration is the number%
of binding edges (i.e., the cardinality of the matching $M$).
The \emph{entropy} $S(\alpha)$ of a configuration is the number of polymers (connected components) of $\alpha$.%

As in \cite{TBNintro}, we study the particularly interesting limiting case in which enthalpy is \emph{infinitely} more favorable than entropy. %
We say a configuration $\alpha$ is \emph{saturated} if it has no pair of domains $d$ and $d^*$ that are both unbound;
this is equivalent to stating that $\alpha$ has maximal bonding among all configurations in $[\vvc]$.
We say a configuration $\alpha\in[\vvc]$ is \emph{stable} (aka thermodynamically favored) if it is saturated and maximizes the entropy among all saturated configurations, i.e., every saturated configuration $\alpha'\in[\vvc]$ obeys $S(\alpha') \leq S(\alpha)$.
Let $[\vvc]_\Box$ denote the set of stable configurations of monomer collection $\vvc$.

\section{Simulating Space-Bounded Turing Machines}
\label{sec:hard-coded-comp}

Our first result proves that for any Turing machine $M$ such that $M$ requires no more than $s$ tape cells and $t$ time steps (we mention time bound $t$ for efficiency of monomer types and polymer size, but the result also holds if we assume the worst case where $t = O(2^s)$), there exists a TBN which simulates $M$.  We provide definitions of what it means for a TBN to simulate a Turing machine, then formally state our Theorem and give our proof, which is by construction.

For the remainder of the section, let $M$ be an arbitrary $s$ space-bounded and $t$ time-bounded Turing machine. Let $i$ be an arbitrary input bit string to $M$, noting that $0 \le |i| \le s$.

Definition~\ref{def:simulate-TM} uses the following notation. 
\begin{enumerate}
\itemsep0em 
\item Let $\calT = (\calD, \calM)$ be a TBN with $\calD$ a finite set of primary domain types and $\calM$ a finite set of monomer types. 
\item Let $\mathcal{O}$ and $\mathcal{I}$ be subsets of $\calM$. We call $\mathcal{O}$ the set of \emph{output monomer types}, and call $\mathcal{I}$ the set of \emph{input monomer types}.
\item Let $\vvc$ be a monomer collection of $\calT$.
\item Let $E_{input}$ be a function from finite sets of input monomers (i.e. monomers of types in $\mathcal{I}$) to binary strings of length $s$. Let $S$ be a finite set of input monomers and let $i$ be a binary string of length $s$. If $E_{input}(S) = i$, then we say that $S$ \emph{encodes} $i$.
\item Let $E_{output}$ be a map from a finite set of output monomers (i.e. monomers of types in $\mathcal{O}$) to binary strings of length $s$. If $S$ is a finite set of monomers and $o\in 2^s$ are such that $E_{output}(S) = o$, then we say that $S$ \emph{encodes} $o$.
\end{enumerate}

Definition~\ref{def:simulate-TM} says that a monomer collection for a TBN  \emph{simulates} a Turing machine $M$ on some input $i$ if 
every stable configuration of the monomer collection is such that every polymer, $p$ say, that contains a monomer with input monomer type contains a set of monomers with input monomer type which encode $i$, and moreover, the set of monomers in $p$ with output monomer type contained in this polymer encode $M(i)$.\footnote{Under some reasonable representation of monomers as binary strings, one might want to require that the encoding is sufficiently weak, in FAC$^0$ for example. However, in this paper we do not require such restrictions on encodings in the definition of ``simulation''. We do note that the encodings $E_{input}$ and $E_{ouput}$ that we use are straightforward encodings which require checking $O(\log s)$ domains in order to determine the input $i$. Moreover, we note that to translate $M$ and $i$ to a TBN, we first translate $M$ and $i$ to an aTAM system via a standard technique and then translate this aTAM system to a TBN using a straightforward approach outlined in Sections~\ref{sec:construction-comps} and~\ref{sec:construction-explosion}.}

\begin{definition}\label{def:simulate-TM}
A monomer collection $\vvc$ for the TBN $\calT$ \emph{simulates} an $s$ space-bounded Turing machine $M$ on input $i$ if and only if
there exist encodings $E_{input}$ and $E_{output}$ such that for every stable configuration $\alpha$ in $[\vvc]_\Box$, if $\alpha$ contains a polymer that contains a monomer with type in $\mathcal{I}$, then 
\begin{enumerate}
\itemsep0em 
\item for $S_{in}$ the set of monomers in $\alpha$ with types in $\mathcal{I}$,  $E_{input}(S_{in}) = i$,
\item for $S_{out}$ the set of monomers in $\alpha$ with types in $\mathcal{O}$,  $E_{output}(S_{out}) = M(i)$, where $M(i)$ is the output of the Turing machine $M$ on input $i$.
\end{enumerate}

\end{definition}

\begin{theorem}\label{thm:hard-coded-sim}
For any $s$ space-bounded, $t$ time-bounded Turing machine $M$, there exists a set of primary domain types $\calD$, and sets of monomer types $\calM$, $\calM_{seed}$, and $\mathcal{O} \subset \calM$ consisting of monomers with binding domains in $\calD\cup \calD^*$
such that, for any valid input $i$ to $M$, the following properties hold.

\begin{enumerate}
    \itemsep0em 
    \item there exists a monomer type $m_i \in \calM_{seed}$ such that $m_i$ encodes $i$, 
    \item for $\calM_i = \calM \cup \{m_i\}$, there exists a monomer collection $\vvc$ for TBN $\calT_i = (\calD, \calM_i)$ such that $\vvc$ simulates $M$ on input $i$, and
    \item the set of output monomer types for the simulation is equal to $\mathcal{O}$.
\end{enumerate} 
\end{theorem}

\subsection{Overview of the proof of Theorem~\ref{thm:hard-coded-sim}}\label{sec:hard-code-overview}

In this section we give a brief overview of the construction used to prove Theorem~\ref{thm:hard-coded-sim}. Please see Section~\ref{sec:hard-coded-comp-append} for the full details.

The abstract Tile Assembly Model (see Section~\ref{sec:tam-informal} for a brief introduction) has been shown to be computationally universal \cite{Winf98}, and many aTAM results utilize a construction technique in which a Turing machine is simulated as a series of columns self-assemble in a zig-zag manner, with each successive column representing the contents of the tape, the state of the machine, and the location of the read-write head at successive steps in time. Despite the fact that the TBN model does not incorporate any geometry or any notion of time, we are still able to leverage the ideas of an aTAM zig-zag system simulating a Turing machine to design the domains and monomers of a TBN.

\begin{figure}[ht]
    \centering
    \includegraphics[width=4in]{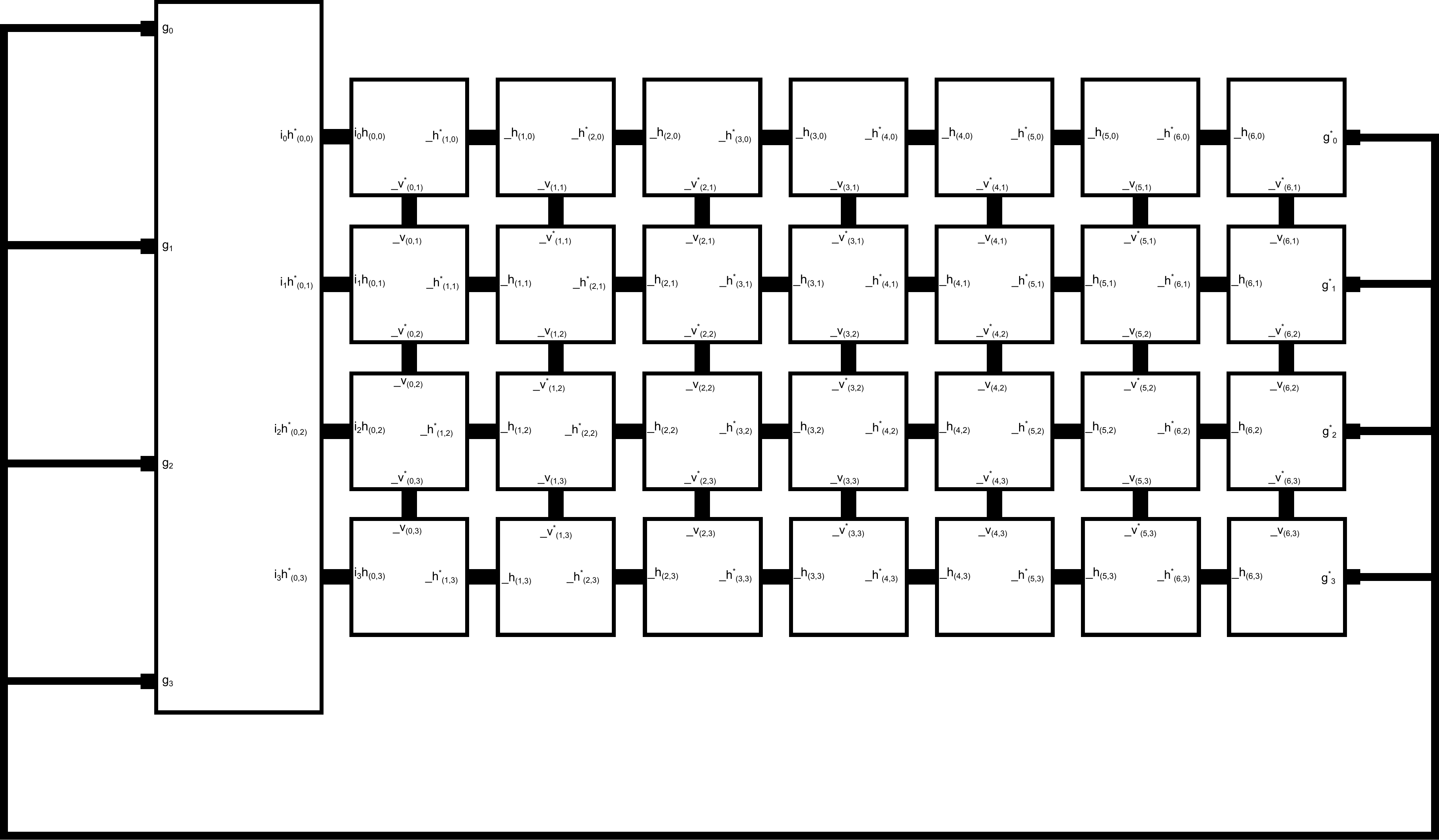}
    \caption{A schematic example of a polymer encoding a Turing machine simulation.}
    \label{fig:hard-coded-overview}
\end{figure}

Essentially, for each tile type of the aTAM system, we first design a monomer with $4$ domains roughly equivalent to the tile's glue labels.  Then, using the time and space bounds of the Turing machine to determine the height of each column and the number of columns, we make a copy of each of those monomers (and their domains) specific for each row and column location. The other main components are the \emph{seed monomer} which has a domain for each bit of input to the Turing machine, \emph{end monomers} which bind to the monomers representing the last column of the computation as well as to the seed monomer, and finally the \emph{cap monomers}.
There is a cap monomer $m_c$ specific to every monomer type $m$ (except the seed monomers), which has exactly one complementary domain for each primary domain type on $m$--- e.g., if $m = \{a, b, c*, d*\}$, the cap monomer for $m$ is $m_c = \{a*,b*\}$. The purpose of cap monomers is to make the attachment of a monomer onto the large TM simulating polymer ``entropy-neutral'': without capping monomers, binding a monomer to the large polymer reduces entropy by one; with capping monomers, the binding of the monomer to the large polymer implies the cap monomer is free, resulting in $0$ net entropy gain/loss.
Additionally, $s$ many end monomers must be bound to the seed monomer in order to maximize enthalpy, where $s$ is the space used by the TM simulation. The complete TM simulating polymer--- by virtue of binding its end monomers to the matching computation monomers as well as the seed --- implies a number of free end cap monomers equal to $s$, creating a net entropy gain of $s$, causing the correct simulating polymer to be in the stable configuration.

Although the number of unique monomer types for this construction is high, an important aspect of it is that, unlike previous results in the TBN model, the construction is robust to inexact counts of monomer types.  In fact, the only requirements for the counts of monomer types in a collection which correctly simulates the Turing machine is that the number of input monomers present is less than or equal to the number of any of the computation monomer or end monomer types, and that the number of each of the cap monomer types is greater than the number of any of the other monomer types.  Given any collection in which the counts of monomer types respect these ratios, that collection correctly simulates the Turing machine (following Definition~\ref{def:simulate-TM}) with $\mathcal{I}$ equal to the input monomer type and $\mathcal{O}$ equal to the end monomer types.  The single stable configuration of any such collection will include (1) a polymer for every copy of the input monomer which contains that input monomer as well as a full set of computation monomers which represent the entire computation and the output encoded by the end monomers (see Figure~\ref{fig:hard-coded-overview} for an example), (2) the leftover, unused computation and end monomers each in a polymer of size $2$ which also includes its unique cap monomer, and (3) the singleton cap monomers whose computation or end monomers are incorporated in the computation-simulating polymers.  Such a configuration is saturated and maximizes entropy over all saturated configurations, and thus is stable.

\section{Simulation of arbitrary Boolean circuits via TBN without a tile assembly pathway}
Under the TBN model--- with no consideration of a corresponding tile assembly system--- the TM simulation discussed in Section~\ref{sec:hard-coded-comp} is easily generalized to arbitrary fan-in fan-out Boolean circuits, mainly via the removal of planarity constraints imposed by the bonds of the aTAM.
One may suspect the removal of an accompanying aTAM system also removes the argument for a plausible kinetic pathway, yet the construction described here is similar enough to the system described in Section~\ref{sec:hard-coded-comp} to argue that a similar ``monomer-by-monomer attachment to a growing assembly'' pathway exists for this construction as well.

An overview of the construction is given here in an explicit example shown in Figure~\ref{fig:circuit}.
More discussion and details may be found in Section~\ref{sec:circuits-append}.
The use of seed monomers, capping monomers, and end monomers are effectively identical to the construction given in Section~\ref{sec:hard-coded-comp}.
The main difference in this construction is the construction of the computation monomers: for each gate $g$ of fan-in $i$ and fan-out $o$ which computes $f_g : \{0,1\}^i \rightarrow \{0,1\}^o$, we construct $2^i$ computation monomers--- one for each possible input to the gate.
The monomer corresponding to a particular input $s$ to the gate exposes domains corresponding to $f_g(s)$ which are complementary with the gate monomers for the gates in $C$ take input from $g$.

 \begin{figure}[ht]
 \centering
     \begin{subfigure}{0.36\textwidth}
         
         \includegraphics[width=1.0\textwidth]{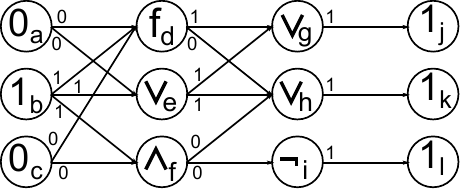}
         \caption{The subscript denotes the gate node; the left-most nodes are input nodes, the right-most nodes are output nodes, and the rest are gate nodes. $f_d(010) = 10$.}\label{fig:circuit_ex}
         \end{subfigure}
     \begin{subfigure}{0.60\textwidth}
         \includegraphics[width=1.0\textwidth]{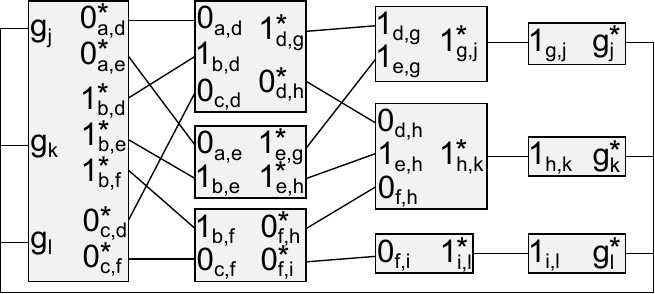}
     \caption{The polymer containing the seed in the stable configuration of the TBN simulating $C$.}
     \label{fig:circuit_polymer_ex}
     \end{subfigure}
 \begin{subfigure}{0.6\textwidth}
     \includegraphics[width=1.0\textwidth]{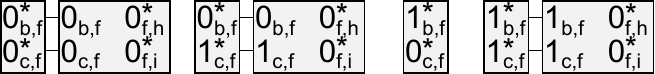}
     \caption{The capping and computation monomers corresponding to gate node f which are not in the polymer containing the seed.}
     \label{fig:circuit_polymer_ex}
     \end{subfigure}
     \caption{An example simulation of an arbitrary Boolean circuit C. (a) shows the circuit represented as a directed acyclic graph with edges corresponding to the input/outputs of the gate. (b) shows the TBN polymer which effectively simulates the circuit. (c) shows an example of the monomer set constructed for each gate in the circuit.}\label{fig:circuit}
 \end{figure}

\vspace{-.5cm}
\section{A negative result on porting computing systems from the aTAM to the TBN model}
\label{sec:aTAM-fail}

In this section we provide very high-level details of an argument which shows what occurs when standard aTAM systems which are designed to simulate Turing machines are treated as TBNs. (Due to space constraints, we present a high-level overview of the argument here, but details can be found in Section \ref{sec:aTAM-fail-details}.) This is similar to the TBN designed for the proof of Theorem~\ref{thm:hard-coded-sim} but without creating unique, hard-coded monomer types for each location.  

Let $M$ be an $s$-space-bounded Turing machine and $T_M$ be a standard zig-zag aTAM tile set which simulates $M$.  (Note that this argument will also apply when $M$ is not space-bounded.) For $n \in N$ where $n$ is a valid input to $M$, let $T_n$ be the set of ``input'' tile types which assemble the binary representation of $n$ as a vertical column to serve as the input to $T_M$.  Then the aTAM system $\calT_{M(n)} = (T_M \cup T_n, \sigma_n, 2)$, where $\sigma_n$ is simply the first tile of $T_n$ at the origin, simulates $M(n)$. We refer to the unique terminal assembly of $\calT_{M(n)}$ as $S_n$ (i.e. simulation $n$).  Additionally, if $X$ is a set of coordinate locations, by $S_n(X)$ we refer to the subassembly of $S_n$ contained at the locations of $X$. Let $i \ne j \ne k$ be valid inputs to $M$, and $S_i$, $S_j$ and $S_k$ be the terminal assemblies of $\calT_{M(i)}$, $\calT_{M(j)}$, and $\calT_{M(k)}$, respectively, such that the following conditions hold:

\begin{enumerate}
  \setlength{\itemsep}{0pt}
  \setlength{\parskip}{0pt}
    \item The outputs $M(i) \ne M(k)$
    \item There exist columns (i.e. sets of all tile locations in a given column) $c_1$ and $c_2$, and individual tile locations $l_1$ and $l_2$ in $c_1$ and $c_2$, respectively, such that:
    \begin{enumerate}\label{cond:matching-locs}
        \item $S_i(c_1) = S_j(c_1)$ (i.e. both columns have the exact same tile types in each location) except at location $l_1$, where they have differing tile types with different glues on their west sides (which would represent different cell values for the respective simulated tape cells of $M$)
        \item $S_j(c_2) = S_k(c_2)$ except at location $l_2$ where they have tile types which differ in their west glues
        \item $S_i(l_1) = S_k(l_2)$
        \item $S_j(l_1) = S_j(l_2)$
        \item $S_i(c_2) \ne S_k(c_2)$
    \end{enumerate}
\end{enumerate}

If a TBN is created to simulate $M$ using the same techniques as for the proof of Theorem~\ref{thm:hard-coded-sim}, these conditions of the computation being simulated on input $i$ allow ``splicing'' to occur between polymers which could represent computations on inputs $i$,$j$, and $k$ and still retain a polymer in which all domains are bound, and a configuration with maximum enthalpy and entropy which does not simulate $M(i)$. (See Figure~\ref{fig:comp-splicing-rewiring} for a schematic depiction.)  The conditions required by this argument are also relatively natural and likely to occur for sets of three inputs for a large number of computations, as they only require that across two pairs of three different inputs there are points at which the computations have nearly identical tape contents, and also a tape cell location whose value is changed at one point and then changed back to the former value later and the rest of the configuration matches across another pair of the two inputs.%
While this result does not show the impossibility of so-called efficient Turing machine simulation, it implies that, if possible, more innovative techniques will be required.  This leads us to our next result, in which we consider a variant of the TBN model which imposes geometry on the monomers and their bonds, and we demonstrate efficient Turing machine simulation. 

\begin{figure}
    \centering
    \includegraphics[height=1.5in]{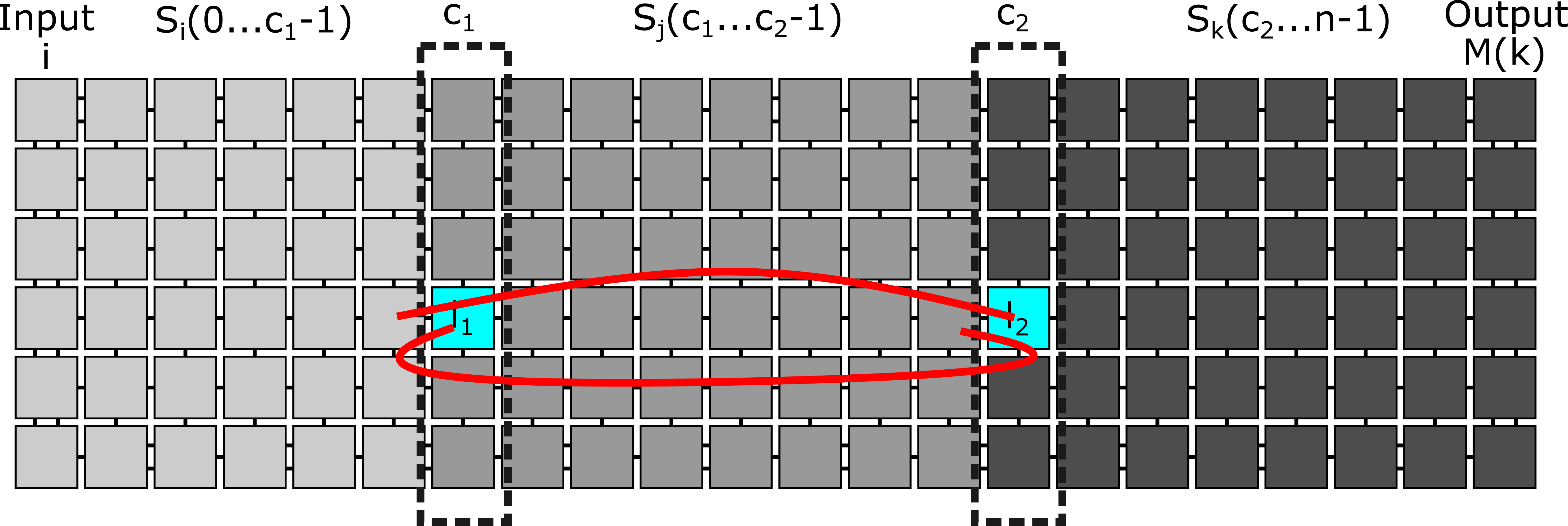}
    \caption{Schematic view of how portions of three computation-simulating polymers could be connected together to yield a polymer representing an invalid computation.}
    \label{fig:comp-splicing-rewiring}
\end{figure}

\section{Geometric Thermodynamic Binding Networks}\label{sec:GTBN}

The result of Section \ref{sec:hard-coded-comp} demonstrates that computation can be simulated by TBNs when they are composed of location-specific monomers for every location within the computation, which is quite inefficient compared to, for instance, the number of unique tile types required to simulate computations within the aTAM. However, Section \ref{sec:aTAM-fail} shows that treating aTAM tiles of computation-simulating systems directly as monomers of a TBN results in systems with many fewer monomer types, but which incorrectly simulate computations. Intuitively, the reason for the failure of such systems is due to the lack of geometry included within the TBN model, which allows for the domains of any monomer to bind to complementary domains of any other monomers, independent of the patterns of connections, which would not be the case if the monomers had to conform to geometric constraints on their sizes and locations.  In order to address this issue, in this section we introduce a refinement to the TBN model which includes such geometric constraints.
\vspace{-.3cm}
\subsection{GTBN Model definition}

We define the \emph{Geometric Thermodynamic Binding Networks} (GTBN) model to be an extension of the TBN model, with a few notable differences which restrict the ways in which polymers can form.  A GTBN is a pair $\calT = (\calD,\calM)$ consisting of a finite set $\calD$ of primary domain types and a finite set $\calM$ of monomer types.  But, rather than a monomer simply being defined as a multiset of domains, a \emph{geometric monomer type} $m \in \calM$ is instead defined as a polygon $p$, along with a set of pairs $(d,l)$ where $d \in \calD \cup \calD^*$ and $l \in \mathbb{R}^2$ is the point on the perimeter of $p$ where $d$ is located.\footnote{Note that the definitions can naturally be extended to 3D polyhedra.}  Geometric monomers are taken to be rigid polygons, and given a pair of geometric monomers, $m_1$ and $m_2$ where $(d_i,l_i) \in m_1$ and $(d_i^*,l_j) \in m_2$, if $m_1$ and $m_2$ can be positioned in the plane so that they do not overlap but the locations $l_i$ and $l_j$ on $m_1$ and $m_2$, respectively, are adjacent to each other, then those domains can bind. Bonds are rigid and therefore so are polymers formed by their binding. Geometric monomers and polymers can be translated and rotated (but not reflected), and can bind together if they can be positioned such that they do not overlap and complementary domains on their perimeters are adjacent.%
In this paper, we will only consider geometric monomers which are unit squares with at most a single domain on any face, located in the center of the face.  (Note that this is similar to tiles in the aTAM, but while the aTAM prevents tiles from rotating through two dimensional space, geometric monomers are allowed to within the GTBN.)  Thus, each monomer in a geometric polymer can be represented by a pair $(p,m)$ where $p \in \mathbb{N}^2$ represents the coordinates of the center of the geometric monomer and $m \in \mathcal{M}$ the monomer type, and a geometric polymer is a set of such pairs, and the \emph{geometric monomer binding graph} contains edges representing complementary domains which are adjacent to each other in some polymer.

A major difference between TBNs and GTBNs is that, due to geometric constraints, it is possible to have a configuration in a GTBN in which there exists an unbound domain $d$ on some monomer and an unbound domain $d^*$ on either that or another monomer, but $d$ and $d^*$ cannot bind together.  That is, it may be impossible for the monomers (or the polymers containing them) to be validly positioned so that the domains are adjacent.  Therefore, we define a condition of a GTBN configuration called \emph{effectively saturated} which occurs when the configuration either (1) is saturated, or (2) for all pairs of domains $d$ and $d^*$ such that both are unbound, there is no valid positioning of the monomers or polymers containing them such that $d$ and $d^*$ can be placed adjacent to each other (i.e. they are geometrically prevented from binding).

\vspace{-.2cm}
\subsection{Efficient simulation of Turing machines in GTBNs}

With the definition of the GTBN model, we are now able to prove that the geometric constraints of the model allow for efficient, accurate simulation of Turing machines.  In this section we present the theorem statement and a high-level overview of the proof, which is by construction.  For full details, please see Section~\ref{sec:GTBN-append}.

\begin{theorem}\label{thm:geo-TM}
Let $L \in DTIME(f(n))$ be a decidable language for arbitrary function $f$, and $M$ be a Turing machine which decides $L$. There exists a set of primary domain types $\mathcal{D}$, and sets of geometric monomer types $\mathcal{M}$, $\mathcal{M}_{seed}$, and $\mathcal{O} \subset \mathcal{M}$ consisting of geometric monomers with binding domains $\mathcal{D} \cup \mathcal{D}^*$ such that, for any valid input $i$ to $M$, the following properties hold:
\begin{enumerate}
    \itemsep0em 
    \item there exists a set of geometric monomer types $m_i \subset \mathcal{M}_{seed}$ such that $m_i$ collectively encodes $i$,
    \item for $\mathcal{M}_i$ = $\mathcal{M} \cup m_i$, there exists a geometric monomer collection $\vvc$ for GTBN $\mathcal{T}_i = (\mathcal{D},\mathcal{M}_i)$ such that $\vvc$ simulates $M$ on input $i$, and
    \item the set of output geometric monomer types for the simulation is equal to $\mathcal{O}$.
\end{enumerate}
\end{theorem}

\begin{figure}
    \centering
    \includegraphics[width=\linewidth]{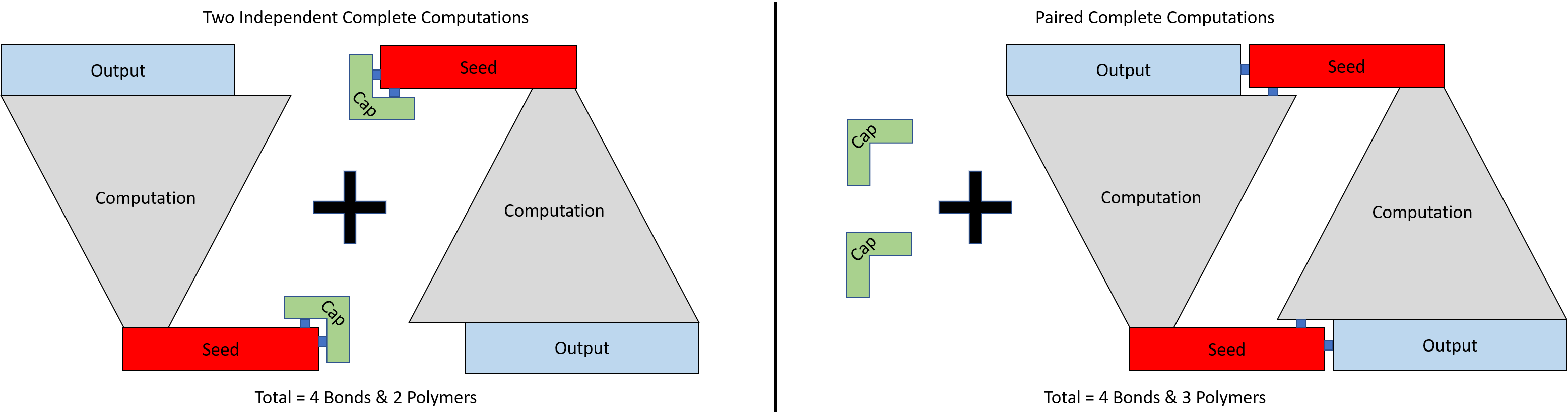}
    \caption{This illustrates the pairing of two complete computations. The dislodging of the seed caps creates an entropy bonus of one. This bonus is what makes the final complete paired construction favorable.}
    \label{fig:paired_computation_gap}
\end{figure}

The proof of Theorem~\ref{thm:geo-TM} is by construction, and begins similarly to the construction for the proof of Theorem~\ref{thm:hard-coded-sim}, with the creation of domains and (geometric) monomers of a GTBN $\calT$ based off of the definition of a zig-zag aTAM system $\calT_M$ which simulates the Turing machine $M$, with a few notable differences and stopping before the need to increase the size of the domain and monomer sets by creating copies hard-coded for each position in the simulation. %
We first note that the geometric monomers are all designed to simply be unit squares like the aTAM tiles, with single domains located in the center of faces to represent the tiles' glues. Since this construction doesn't require monomers hard-coded to locations, and in fact doesn't require a fixed number of rows or columns, it is able to simulate a tape of steadily increasing length and so utilizes collections of monomers that combine to extend the length of the tape. 

It is ensured that the deterministic path followed by the zig-zag aTAM system which simulates the same Turing machine is faithfully encoded by the resulting ``computation'' polymer of the single stable configuration by the geometric constraints placed on the positioning of geometric monomers and the rigidity of their bonds, which prevents erroneous ``re-wiring'' to occur as it could in the regular TBN example of Section~\ref{sec:aTAM-fail}.  However, in order to create an entropy gap which makes the configuration containing the correct computation simulations the single stable configuration, since we can no longer have bound domains which span the full distance of the polymer (as they do from the seed to end monomers in the proof of Theorem~\ref{thm:hard-coded-sim}), we instead provide an analogous method of freeing additional caps--- thus gaining entropy--- by designing the monomers so that polymers encoding the computation combine in pairs (as seen in Figure~\ref{fig:paired_computation_gap}).

Thus, an arbitrary halting Turing machine computation can be simulated efficiently in terms of domain and monomer type counts, both of which are $O(|Q||\Gamma|)$ (where $Q$ is the state set and $\Gamma$ is the tape alphabet).  As with the construction for the proof of Theorem~\ref{thm:hard-coded-sim}, this construction is robust over a class of configurations in which relationships exist between the counts of different categories of monomers. The inclusion of the fact that the language being decided $L \in DTIME(f(n))$ is simply to specify the count of computation monomers which must be included in the collection, relative to input seeds, to ensure that the computation can be completely represented without running out of monomers, i.e. $O(f(n)^2)$ copies of the computation monomers must be available per copy of the seed monomer.

\vspace{-5pt}
\section*{Acknowledgements}
\vspace{-5pt}
The authors would like to thank David Doty, Pierre{-}{\'{E}}tienne Meunier, David Soloveichik, Chris Thachuk, and Damien Woods for much help in developing these results.
\vspace{-10pt}

\bibliographystyle{plain}
\bibliography{tam,experimental_refs,tbn}

\begin{thebibliography}{10}

\bibitem{arxivTBNstable}
Keenan Breik, Lakshmi Prakash, Chris Thachuk, Marijn Heule, and David
  Soloveichik.
\newblock Computing properties of stable configurations of thermodynamic
  binding networks.
\newblock {\em CoRR}, abs/1709.08731, 2017.

\bibitem{Versus}
Sarah Cannon, Erik~D. Demaine, Martin~L. Demaine, Sarah Eisenstat, Matthew~J.
  Patitz, Robert~T. Schweller, Scott~M. Summers, and Andrew Winslow.
\newblock Two hands are better than one (up to constant factors): Self-assembly
  in the 2ham vs. atam.
\newblock In Natacha Portier and Thomas Wilke, editors, {\em STACS}, volume~20
  of {\em LIPIcs}, pages 172--184. Schloss Dagstuhl - Leibniz-Zentrum fuer
  Informatik, 2013.

\bibitem{AGKS05g}
Qi~Cheng, Gagan Aggarwal, Michael~H. Goldwasser, Ming-Yang Kao, Robert~T.
  Schweller, and Pablo~Moisset de~Espan\'{e}s.
\newblock Complexities for generalized models of self-assembly.
\newblock {\em SIAM Journal on Computing}, 34:1493--1515, 2005.

\bibitem{CookFuSch11}
Matthew Cook, Yunhui Fu, and Robert~T. Schweller.
\newblock Temperature 1 self-assembly: Deterministic assembly in 3{D} and
  probabilistic assembly in 2{D}.
\newblock In {\em SODA 2011: Proceedings of the 22nd Annual ACM-SIAM Symposium
  on Discrete Algorithms}. SIAM, 2011.

\bibitem{DotCACM}
David Doty.
\newblock Theory of algorithmic self-assembly.
\newblock {\em Commun. ACM}, 55(12):78--88, December 2012.

\bibitem{jLSAT1}
David Doty, Matthew~J. Patitz, and Scott~M. Summers.
\newblock Limitations of self-assembly at temperature 1.
\newblock {\em Theoretical Computer Science}, 412:145--158, 2011.

\bibitem{TBNintro}
David Doty, Trent~A. Rogers, David Soloveichik, Chris Thachuk, and Damien
  Woods.
\newblock Thermodynamic binding networks.
\newblock In {\em {DNA} Computing and Molecular Programming - 23rd
  International Conference, {DNA} 23, Austin, TX, USA, September 24-28, 2017,
  Proceedings}, pages 249--266, 2017.

\bibitem{evans2014crystals}
Constantine~Glen Evans.
\newblock {\em Crystals that count! Physical principles and experimental
  investigations of DNA tile self-assembly}.
\newblock PhD thesis, California Institute of Technology, 2014.

\bibitem{Polyominoes}
S{\'a}ndor~P. Fekete, Jacob Hendricks, Matthew~J. Patitz, Trent~A. Rogers, and
  Robert~T. Schweller.
\newblock Universal computation with arbitrary polyomino tiles in
  non-cooperative self-assembly.
\newblock In {\em Proceedings of the Twenty-Sixth Annual ACM-SIAM Symposium on
  Discrete Algorithms (SODA 2015), San Diego, CA, USA {\rm January 4-6, 2015}},
  pages 148--167, 2015.

\bibitem{OptimalShapes3D}
David Furcy, Samuel Micka, and Scott~M. Summers.
\newblock Optimal program-size complexity for self-assembly at temperature 1 in
  3d.
\newblock In {\em {DNA} Computing and Molecular Programming - 21st
  International Conference, {DNA} 21, Boston and Cambridge, MA, USA, August
  17-21, 2015. Proceedings}, pages 71--86, 2015.

\bibitem{Polygons}
Oscar Gilber, Jacob Hendricks, Matthew~J. Patitz, and Trent~A. Rogers.
\newblock Computing in continuous space with self-assembling polygonal tiles.
\newblock In {\em Proceedings of the Twenty-Seventh Annual ACM-SIAM Symposium
  on Discrete Algorithms (SODA 2016), Arlington, VA, USA {\rm January 10-12,
  2016}}, pages 937--956, 2016.

\bibitem{FlexibleVsRigid}
Nata\v{s}a Jonoska and GregoryL. McColm.
\newblock Flexible versus rigid tile assembly.
\newblock In CristianS. Calude, MichaelJ. Dinneen, Gheorghe Păun, Grzegorz
  Rozenberg, and Susan Stepney, editors, {\em Unconventional Computation},
  volume 4135 of {\em Lecture Notes in Computer Science}, pages 139--151.
  Springer Berlin Heidelberg, 2006.

\bibitem{jCCSA}
James~I. Lathrop, Jack~H. Lutz, Matthew~J. Patitz, and Scott~M. Summers.
\newblock Computability and complexity in self-assembly.
\newblock {\em Theory Comput. Syst.}, 48(3):617--647, 2011.

\bibitem{WoodsMeunierSTOC}
Pierre{-}{\'{E}}tienne Meunier and Damien Woods.
\newblock The non-cooperative tile assembly model is not intrinsically
  universal or capable of bounded turing machine simulation.
\newblock In {\em Proceedings of the 49th Annual {ACM} {SIGACT} Symposium on
  Theory of Computing, {STOC} 2017, Montreal, QC, Canada, June 19-23, 2017},
  pages 328--341, 2017.

\bibitem{jSignals}
Jennifer~E. Padilla, Matthew~J. Patitz, Robert~T. Schweller, Nadrian~C. Seeman,
  Scott~M. Summers, and Xingsi Zhong.
\newblock Asynchronous signal passing for tile self-assembly: Fuel efficient
  computation and efficient assembly of shapes.
\newblock {\em International Journal of Foundations of Computer Science},
  25(4):459--488, 2014.

\bibitem{PatitzSurvey}
Matthew~J. Patitz.
\newblock An introduction to tile-based self-assembly and a survey of recent
  results.
\newblock {\em Natural Computing}, 13(2):195--224, 2014.

\bibitem{jSADS}
Matthew~J. Patitz and Scott~M. Summers.
\newblock Self-assembly of decidable sets.
\newblock {\em Natural Computing}, 10(2):853--877, 2011.

\bibitem{RotWin00}
Paul W.~K. Rothemund and Erik Winfree.
\newblock The program-size complexity of self-assembled squares (extended
  abstract).
\newblock In {\em STOC '00: Proceedings of the thirty-second annual ACM
  Symposium on Theory of Computing}, pages 459--468, Portland, Oregon, United
  States, 2000. ACM.

\bibitem{SchulmanWinfreeNucleation}
Rebecca Schulman and Erik Winfree.
\newblock Programmable control of nucleation for algorithmic self-assembly.
\newblock {\em SIAM J. Comput.}, 39(4):1581--1616, 12/2009.

\bibitem{SchWin07}
Rebecca Schulman and Erik Winfree.
\newblock Synthesis of crystals with a programmable kinetic barrier to
  nucleation.
\newblock {\em Proceedings of the National Academy of Sciences},
  104(39):15236--15241, 2007.

\bibitem{SolWin07}
David Soloveichik and Erik Winfree.
\newblock Complexity of self-assembled shapes.
\newblock {\em SIAM Journal on Computing}, 36(6):1544--1569, 2007.

\bibitem{Winf98}
Erik Winfree.
\newblock {\em Algorithmic Self-Assembly of {D}{N}{A}}.
\newblock PhD thesis, California Institute of Technology, June 1998.

\bibitem{WoodsIU2013}
Damien Woods.
\newblock Intrinsic universality and the computational power of self-assembly.
\newblock In {\em MCU: Proceedings of Machines, Computations and Universality},
  volume 128, pages 16--22, Univ.\ of Z\"{u}rich, Switzerland. Sept. 9-12,
  2013. Open Publishing Association.

\end{thebibliography}

\clearpage
\appendix

\section{Informal description of the abstract Tile Assembly Model}
\label{sec:tam-informal}

This section gives a brief informal sketch of the abstract Tile Assembly Model (aTAM), which will be referred to in several sections. See \cite{RotWin00} for a formal definition of the aTAM.

A \emph{tile type} is a unit square with four sides, each consisting of a \emph{glue label}, often represented as a finite string, and a nonnegative integer \emph{strength}. A glue~$g$ that appears on multiple tiles (or sides) always has the same strength~$s_g$. %
There are a finite set $T$ of tile types, but an infinite number of copies of each tile type, with each copy being referred to as a \emph{tile}. An \emph{assembly}
is a positioning of tiles on the integer lattice $\Z^2$, described  formally as a partial function $\alpha:\Z^2 \dashrightarrow T$. %
Let $\mathcal{A}^T$ denote the set of all assemblies of tiles from $T$, and let $\mathcal{A}^T_{< \infty}$ denote the set of finite assemblies of tiles from $T$.
We write $\alpha \sqsubseteq \beta$ to denote that $\alpha$ is a \emph{subassembly} of $\beta$, which means that $\dom\alpha \subseteq \dom\beta$ and $\alpha(p)=\beta(p)$ for all points $p\in\dom\alpha$.
Two adjacent tiles in an assembly \emph{interact}, or are \emph{attached}, if the glue labels on their abutting sides are equal and have positive strength. %
Each assembly induces a \emph{binding graph}, a grid graph whose vertices are tiles, with an edge between two tiles if they interact.
The assembly is \emph{$\tau$-stable} if every cut of its binding graph has strength at least~$\tau$, where the strength   of a cut is the sum of all of the individual glue strengths in the cut.

A \emph{tile assembly system} (TAS) is a triple $\calT = (T,\sigma,\tau)$, where $T$ is a finite set of tile types, $\sigma:\Z^2 \dashrightarrow T$ is a finite, $\tau$-stable \emph{seed assembly},
and $\tau$ is the \emph{temperature}.
An assembly $\alpha$ is \emph{producible} if either $\alpha = \sigma$ or if $\beta$ is a producible assembly and $\alpha$ can be obtained from $\beta$ by the stable binding of a single tile.
In this case we write $\beta\to_1^\calT \alpha$ (to mean~$\alpha$ is producible from $\beta$ by the attachment of one tile), and we write $\beta\to^\calT \alpha$ if $\beta \to_1^{\calT*} \alpha$ (to mean $\alpha$ is producible from $\beta$ by the attachment of zero or more tiles).
When $\calT$ is clear from context, we may write $\to_1$ and $\to$ instead.
We let $\prodasm{\calT}$ denote the set of producible assemblies of $\calT$.
An assembly is \emph{terminal} if no tile can be $\tau$-stably attached to it.
We let   $\termasm{\calT} \subseteq \prodasm{\calT}$ denote  the set of producible, terminal assemblies of $\calT$.
A TAS $\calT$ is \emph{directed} if $|\termasm{\calT}| = 1$. Hence, although a directed system may be nondeterministic in terms of the order of tile placements,  it is deterministic in the sense that exactly one terminal assembly is producible (this is analogous to the notion of {\em confluence} in rewriting systems).

\section{Zig-zag tile assembly systems}\label{sec:zig-zag}

In this section we describe the general structure of a commonly used type of aTAM tile assembly system for simulating the behavior of Turing machines.  Modified versions of such systems will be used as the basis for our constructions.

A \emph{zig-zag} aTAM system is one which grows in a strict column-by-column (or row-by-row) ordering.  More specifically, the first column grows either bottom-to-top or top-to-bottom, completely, at which point the second column begins growth in the opposite direction.  When it completes, the third grows, again in reversed direction, and so on.

\begin{figure}[ht]
    \centering
    \includegraphics[height=1.5in]{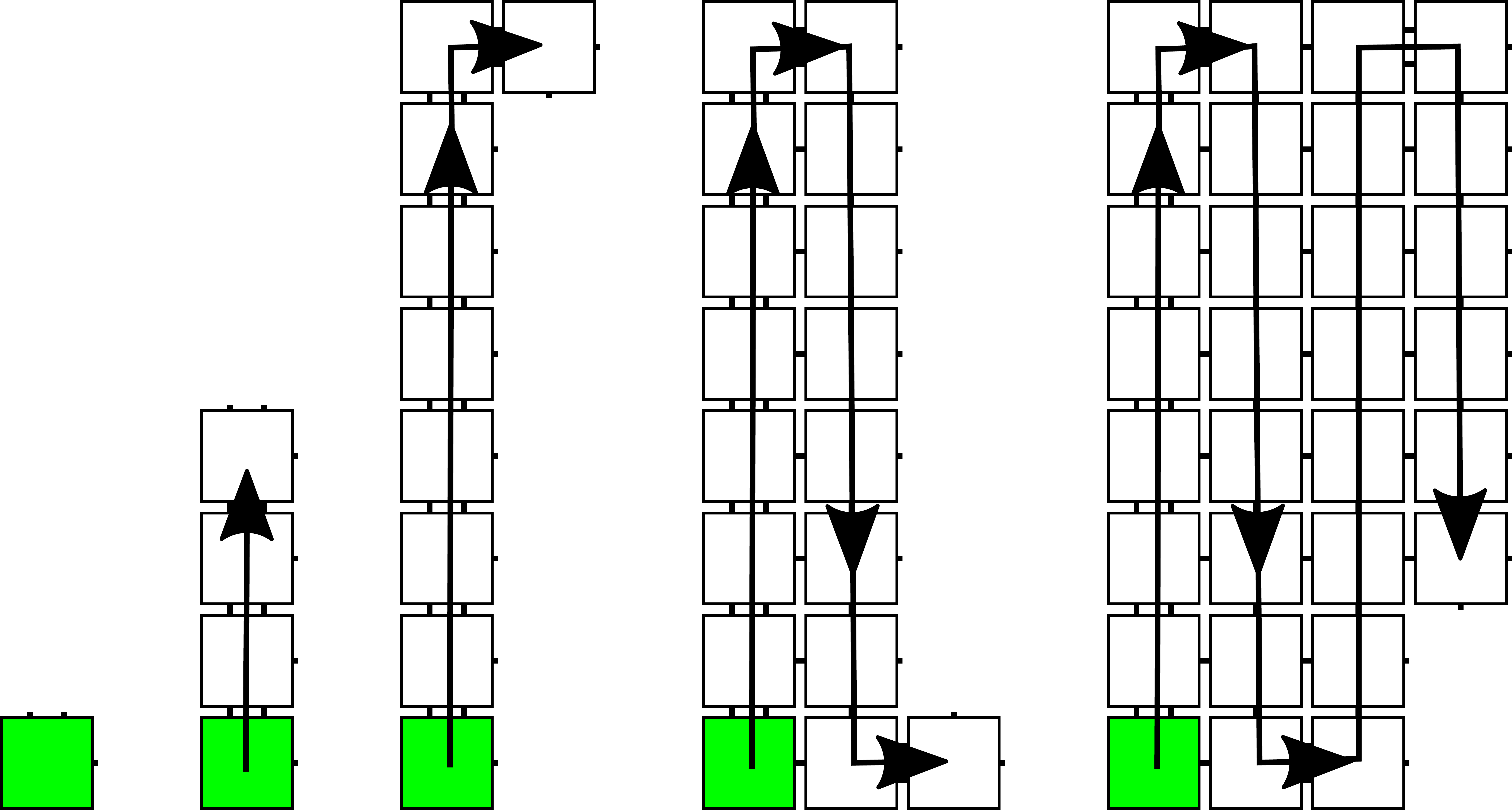}
    \caption{An example of zig-zag growth in the aTAM; the seed tile is shaded green.}
    \label{fig:zig-zag-example}
\end{figure}

Let $M = (Q, \Sigma, \Gamma, \delta, q_0, q_H)$ be an arbitrary Turing machine with state set $Q$, input alphabet $\Sigma$, tape alphabet $\Gamma$, transition function $\delta:(Q,\Gamma) \rightarrow (Q,\Gamma,D)$ (where $D$ is the set of directions $\{L,R\}$), start state $q_0$, and halting state $q_H$. It was shown in \cite{CookFuSch11} that for an arbitrary Turing machine $M$ there exists a zig-zag aTAM system which, for input $i$ to $M$, simulates $M(i)$.  %
We will first discuss the simulation of a space-bounded Turing machine, then explain how it can be modified for artibtrary Turing machines.  In general, a zig-zag aTAM system which simulates an $s$-space-bounded Turing machine $M$ on input $i$ works as follows, with the convention that each column represents a complete encoding of $M$'s tape (with each tile representing a tape cell) as well as the location of the read/write head and $M$'s current state.  First, a column of tiles which encode the bits of $i$ via their east-facing glues, along with enough additional tiles to represent blank tape cells for the remaining amount of tape space (i.e. $s - |i|$) self-assembles. Additionally, exactly one of the east-facing glues encoding an input bit also encodes the start state of $M$, and this is how the location of the read/write head is represented, as well as $M$'s state. We call these tiles the \emph{input tiles}. (For an example, see the left section of Figure~\ref{fig:zig-zag-transition} which shows the initial column of a simulation with input $1011$, space bound $6$, the tape head on the leftmost - topmost in this orientation - input symbol, and the initial state being state $A$.) As each subsequent column grows, in alternating order, most tiles simply copy the value of their corresponding tape cells forward to the right. (We call these tiles the \emph{cell copying tiles}.)  However, at the location representing the tape head, if $M$'s next transition moves the tape head in the same direction in which that column is growing, the transition is simulated.  This means that the value of the tape cell is changed if necessary, and the vertical glue encodes the state of the transition into which $M$ transitions.  (This can be seen in the middle portion of Figure~\ref{fig:zig-zag-transition} which simulates transition $(A,1) \rightarrow (A,1,R)$.) If the column happens to be growing in the direction opposite in which the tape head needs to move, no transition is simulated until the next column forms, as it will then be growing in the correct direction. We call the tiles that encode the tape cell head and $M$'s state the \emph{transition tiles}. (An example can be seen in the right portion of Figure~\ref{fig:zig-zag-transition}.)  At the top and bottom of the columns, the glues encode not only the tape cell value, but also the information that they are at the end of a column, and when they are the last tile to be placed in a column they are designed with a strength-2 glue which initiates the growth of the next column.  We call these tiles the \emph{column end tiles}.

\begin{figure}[ht]
    \centering
    \includegraphics[height=2.5in]{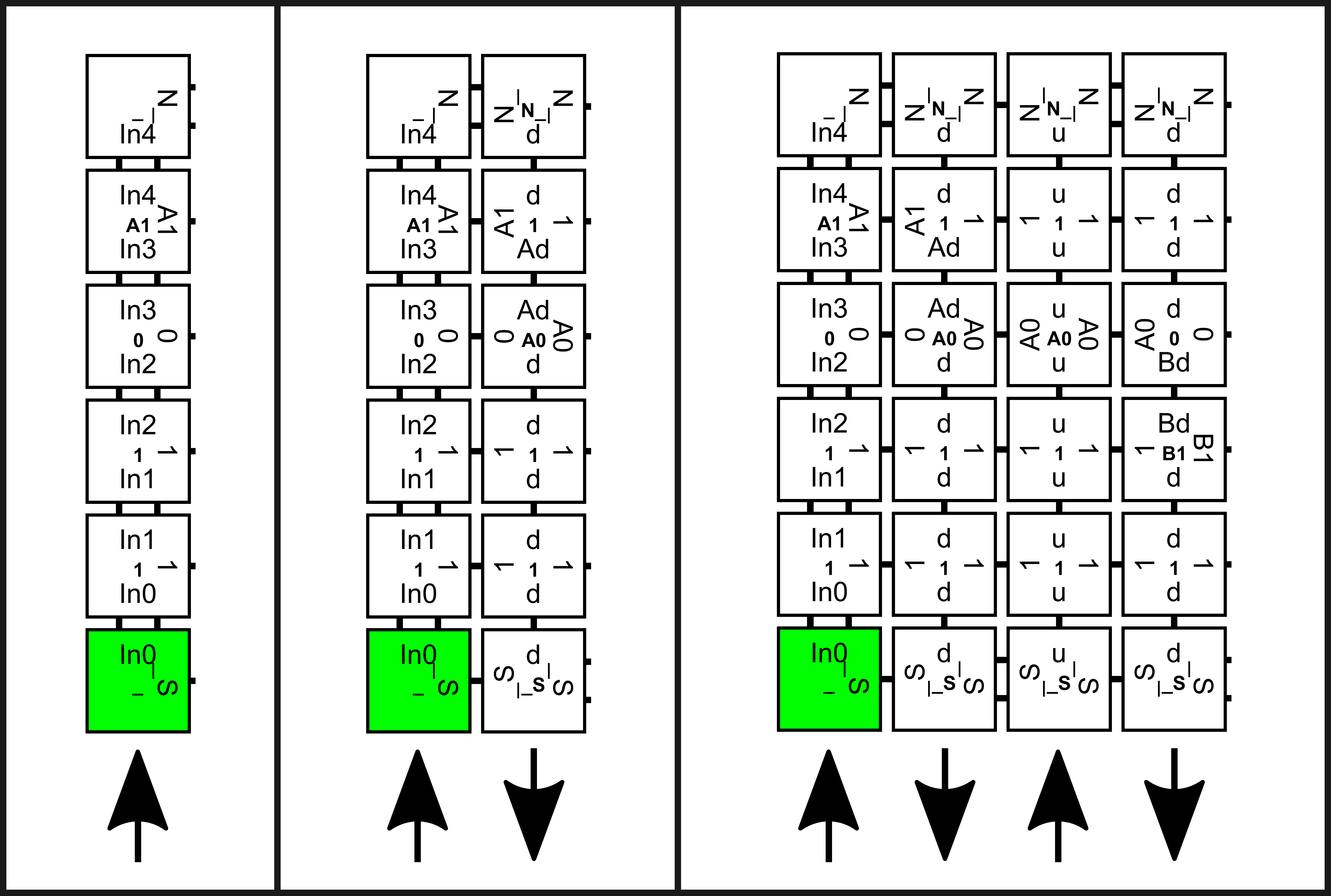}
    \caption{An example of zig-zag growth in the aTAM.}
    \label{fig:zig-zag-transition}
\end{figure}

In order to perform a simulation which is not space bounded, each column simply extends its height by adding one additional tile (which represents a new blank tape cell) at the end of the column beyond the height of the previous column before turning to begin the next.  This produces a tape which add cells more quickly than they can be visited by the tape head, effectively simulating an unbounded tape.  In either case, bounded or unbounded, once the simulation reaches a transition in which $M$ transitions into a halting state, the currently forming column completes and no further tiles attach. (We call the tiles in the last column the \emph{output tiles}.) Also note that, in either case, the \emph{tile complexity}, i.e. the number of unique tile types required, for such a system is simply $O(|Q||\Gamma|)$, which is quite efficient as it is independent of the time or space requirements of the computation.

\section{Technical Details of Simulating Space-Bounded Turing Machines}\label{sec:hard-coded-comp-append}

\begin{proof}
We prove Theorem~\ref{thm:hard-coded-sim} by construction. Let $M = (Q, \Sigma, \Gamma, \delta, q_0, q_H)$ be an arbitrary $s$ space-bounded, $t$ time-bounded Turing machine to simulate, with state set $Q$, input alphabet $\Sigma$, tape alphabet $\Gamma$, transition function $\delta:(Q,\Gamma) \rightarrow (Q,\Gamma,D)$ (where $D$ is the set of directions $\{L,R\}$), start state $q_0$, and halting state $q_H$. We will define $\calT_i$ based on the definition of $M$, as well as the space and time bounds $s$ and $t$, by creating groups of monomers which are logically grouped by functionality into a few main components, which we will now present.

\subsection{Construction components}\label{sec:construction-comps}

The monomers of $\calT_i$ can be logically grouped into four categories. Examples of each these categories monomers can be seen in Figures~\ref{fig:capping1} through~\ref{fig:capping4}. We now describe each monomer category.

\begin{enumerate}
    \item The \emph{seed monomers}:   Referred to as $m_i \in \mathcal{M}_{seed}$ for each valid input $i$, let $i_n$ be the value of the $n$th bit of $i$. For each value of $0 \le n < |i|$, define domain type $d_n = i_nh_{(0,n)}$ and add $d_n$ to $\mathcal{D}$.  We construct monomer $m_i$ by including a domain of the codomain type for each of those domain types.  To $m_i$ we also add ``$g$'' domains which bind to the $\mathcal{M}_{end}$ monomers (defined in~\ref{item:end-monomer}), and we call these \emph{helper domains}. For each value of $0 \le n < |i|$, define domain type $g_n$ and add $g_n$ to $\mathcal{D}$ and a copy of $g_n$ to $m_i$.

    \item The \emph{computation monomers}: We'll call this group of monomers $\mathcal{M}_{comp}$.   Let $\calT_M = (T_M, \sigma_i, 2)$ be a zig-zag aTAM system which simulates $M(i)$. (Please see Section~\ref{sec:zig-zag} for a description of how $\calT_M$ is created for $M$.  Also note that we augment the definition of $\calT_M$ such that once $M$ enters the halting state $q_H$, $\calT_M$ continues growing columns which keep the head in its current location and all tape call values unchanged, essentially performing ``no-op'' transitions.\footnote{This convention is for ease of discussion related to how the computation-simulating polymer completes, ensuring that it grows through the full time bound.  However, for more efficient monomer usage but slightly higher monomer type count, the system could be designed so that the polymer growth stops at the actual length of each specific computation.} For each $t \in T_M$, if $t$ attaches in a column which is growing from bottom to top, we consider the \emph{input} sides to be the south and west, and if $t$ attaches in a column which is growing from top to bottom, the input sides are its north and west (since these are the sides which bind as it first attaches).  Add the labels of $t$'s input glues to the set of primary domain types $\mathcal{D}$, but append $v$ if it's a north or south glue, and $h$ if it's a west glue. The strengths of the glues are disregarded.  Note that for each \emph{output} (i.e. non-input) glue of $t$, those will be input glues of some other tile types and thus will be added to $\mathcal{D}$, and $\mathcal{D}^*$ will contain the complement of everything in $\mathcal{D}$.  (The first two steps of Figure~\ref{fig:domain-renaming} depict this transformation.) 
    
\begin{figure}[ht]
    \centering
    \includegraphics[width=5in]{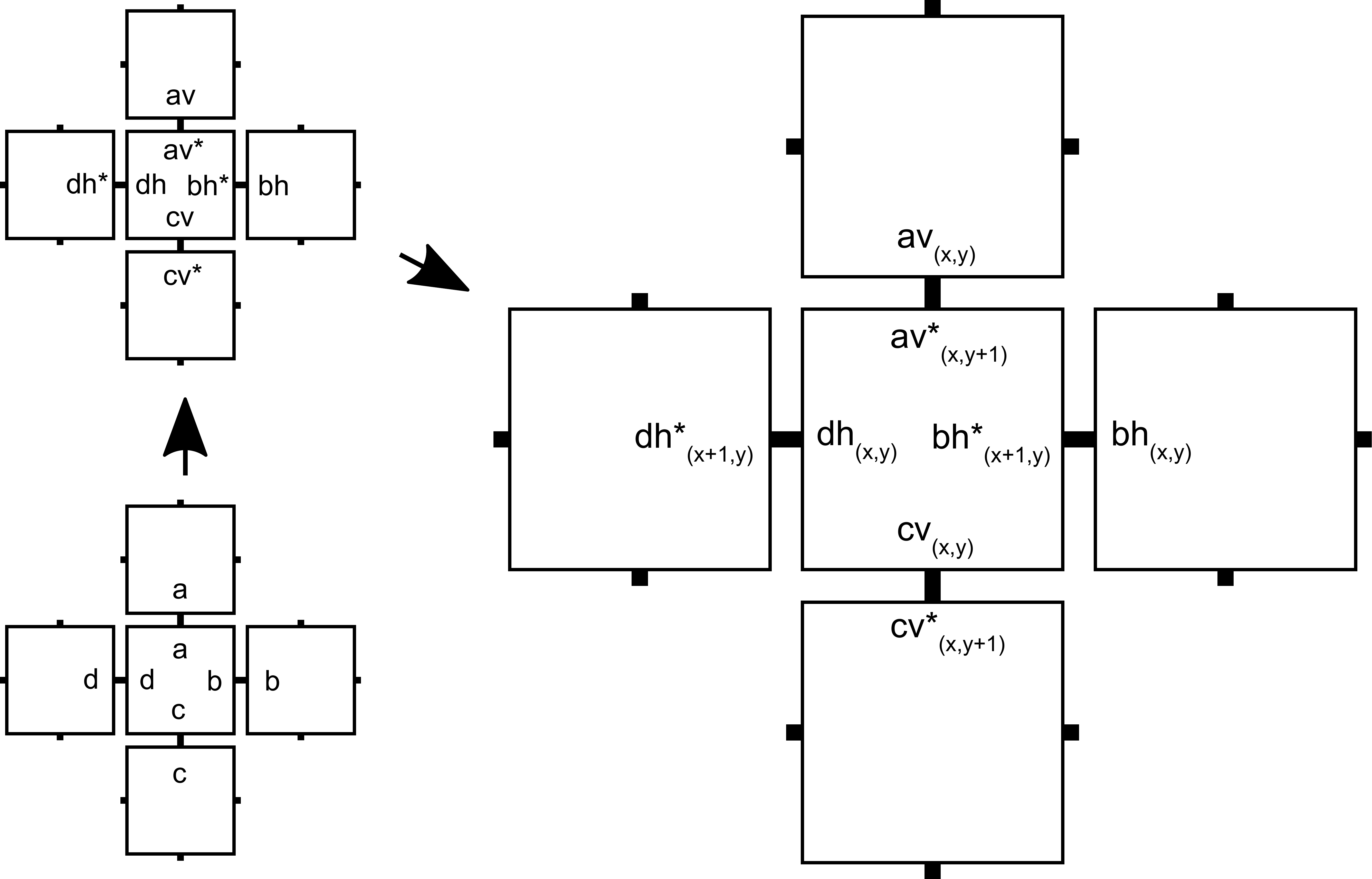}
    \caption{Depiction of how the domains for a computation monomer are created from the glue labels of a computation tile (center of each) for a column which grows bottom-to-top (i.e. has input sides on the west and south). First, ``$h$'' is appended to horizontal glues and ``$v$'' to vertical glues, and output glues have ``$*$'' appended to designate them as codomains.  Next, to make a hard-coded version for each row and column location, the appropriate $x$ and $y$ coordinates are added as subscripts so they can bind to the correct neighbors.}
    \label{fig:domain-renaming}
\end{figure}

    For each $t \in T_M$ where $t$ attaches in a column which grows bottom-to-top, if it has glues $g_n,g_e,g_s,g_w$ as its north, east, south, and west glues, respectively,  define monomer type $C_t \in \mathcal{M}_{comp}$ as $(g_nv*,g_ev*,g_sv,g_wv)$, i.e. using primary domain types to represent the input (south and west glues) and codomain types to represent output (north and east glues). For each $t \in T_M$ where $t$ attaches in a column which grows top-to-bottom, the input changes from south to north and thus $C_t \in \mathcal{M}_{comp}$ is defined as $(g_nv,g_ev*,g_sv*,g_wv)$.

    \item The \emph{end monomers}: We'll call this group of monomers $\mathcal{M}_{end}$. For each possible value  $s \in \{0,1,\_,q_H0,q_H1,q_H\_\}$ in the locations of the final computation column, there exists $m_{end,s_n} \in \mathcal{M}_{end}$ which has domains $s$, $g^*$, $up$, and $down^*$ (corresponding to input sides north and west since these will correspond to a column which grows top-to-bottom because there will always be an even number of computation columns in this construction and the first grows bottom to top).  These will be the monomers which attach to the final column of computation monomers which simulate $M(i)$, represent the output of the computation $M(i)$ since $\mathcal{O} = \mathcal{M}_{end}$, and also bind (via the $g^*$ domains after they are subscripted for their position in the next section) to the seed monomer $m_i$. \label{item:end-monomer}

    \item The \emph{capping monomers}: The set of capping monomers, $\mathcal{M}_{cap}$, contains one monomer $G_k \in \mathcal{M}_{cap}$ for each computation or end monomer $C_k \in \mathcal{M}_{comp} \cup \mathcal{M}_{end}$. Each $G_k$ contains one domain which is complementary to each input domain on $C_k$, i.e. capping monomer $G_k$ has exactly two domains and they are complementary to the two input domains of monomer $C_k$.%
    We call $G_k$ the \emph{matching} capping monomer for $C_k$. 
\end{enumerate}

\subsection{Hard-coding locations}\label{sec:construction-explosion}

This construction ensures that for each monomer type which binds in the polymer simulating the computation, there is only one copy in that polymer.  This means that for each location of a tile in the terminal assembly of $\calT_M$, we must make a copy of each monomer type of $\mathcal{M}_{comp}$ which is unique for that location.  (Note that we could be a bit more selective about how we create these unique copies since not every tile type has a chance to go in each location, e.g. those for rows which grow top-to-bottom could never appear in a row growing bottom-to-top.  However, this does not change the monomer complexity asymptotically, so for ease of discussion we will perform the more naive approach.)

Given the space bound $s$ and time bound $t$, and noting that in the worst case it could take two columns of growth to simulate a single transition of $M$ (i.e. if the next transition needs to move the head in the opposite direction that the current column is growing, the column will complete without performing a transition and the next column will simulate the transition), we will need to create a $2t \times s$ grid. Therefore let $\mathcal{D}_{xy} = \{d_{(x,y)} \vert d \in \mathcal{D}, x \in \{0,1,...,2t-1\}, y \in \{0,1,...,s-1\}\}$, i.e. for each domain $d \in \mathcal{D}$, the new set $\mathcal{D}_{xy}$ contains a copy of the original domain subscripted with $(x,y)$ for each $(x,y)$ coordinate pair.  Then for each $m = (d_1v,d_2h,d_3v,d_4h) \in \mathcal{M}_{comp}$, $x \in \{0,1,...,2t-1\}$, and $y \in \{0,1,...,s-1\}$, $\mathcal{M}_{comp,xy}$ contains $m_{xy} = (d_1v_{(x,y)},d_2h_{(x,y)},d_3v_{(x,y+1)}, d_4h_{(x+1,y)})$.  (The last step of Figure~\ref{fig:domain-renaming} depicts this transformation.)  Additionally, the monomers of $\mathcal{M}_{end}$ are expanded to $\mathcal{M}_{end,xy}$ so that there is a copy for each domain other than $g^*$ subscripted with the x-coordinate $2t-1$ (to attach to the east of the final computation column) and y-coordinate corresponding to each of the rows, and the $g^*$ subscripted just with the row values (to yield codomains $g^*_n$ for $0 \le n < |i|$, allowing them to bind to the domains of $m_i$)..

Similarly, the full set of capping monomers, $\mathcal{M}_{cap,xy}$, is created by making a monomer type for each $m \in \mathcal{M}_{cap}$ in each $(x,y)$ position.

The final set $\mathcal{M} = \mathcal{M}_{comp,xy} \cup \mathcal{M}_{end,xy} \cup \mathcal{M}_{cap,xy}$.

\subsection{Proof of Theorem~\ref{thm:hard-coded-sim}}\label{sec:hard-code-proof}

\subsubsection{Notation}\label{sec:hard-coded-notation}

In this section we introduce some notation used throughout the proof.

Let $\#(\vvc,m)$ be the count of the monomers of type $m$ in $\vvc$, and  $c_{min} = min\{\#(\vvc,m) | m \in \mathcal{M}_{comp, xy} \cup \mathcal{M}_{end, xy}\}$ (i.e. the number of monomers with computation or end monomer type with the fewest number of monomers in $\vvc$), and $k_{min} = min\{\#(\vvc,m) | m \in \mathcal{M}_{cap, xy}\}$. Let $\#comp(\vvc) = \Sigma_{m\in \calM_{comp, xy}} \#(\vvc, m)$, and let $\#end(\vvc) =  \Sigma_{m\in \calM_{end, xy}} \#(\vvc, m)$.

We call monomers with type in $\calM_{comp,xy}$ \emph{computation monomers}, monomers with type in $\calM_{cap,xy}$ \emph{capping monomers}, monomers with type in $\calM_{end,xy}$ \emph{end monomers}, and monomers with type in $\calM_{seed}$ \emph{seed monomers},

Note that by construction of $\calM$, domains and codomains have been encoded with pairs of integers $(x,y)$ corresponding to locations in $\Z^2$. We define the partial function $loc$ with domain in $\calD\cup \calD^*$ and range in $\Z^2$ to be such that $loc(d) = (x,y)$ if the domain or codomain $d$ encodes $(x,y)$ and undefined otherwise. Domains and codomains also encode a horizontal or vertical orientation using the symbols $h$ and $v$. We define a partial function $or$ with domain in $\calD\cup \calD^*$ and range in $\{h,v\}$ to be such that $or(d) = h$ if $d$ encodes $h$, $or(d) = v$ if $d$ encodes $v$, and undefined otherwise. For a computation monomer $m$, we let $loc(m) = (x,y)$ iff $m$ contains a domain $d$ with $loc(d) = (x,y)$.

For a polymer $p$, let $x_{max}(p) \in \N$ be such that there exists $y\in \N$ and a domain $d$ belonging to a monomer in $p$ with $loc(d) = (x_{max}(p), y)$, and moreover, for any domain $d'$ belonging to a monomer in $p$, if $loc(d') = (x', y')$, then $x' \le x_{max}(p)$. Similarly, let $y_{max}(p) \in \N$ be such that there exists $x\in \N$ and a domain $d$ belonging to a monomer in $p$ with $loc(d) = (x, y_{max}(p))$, and moreover, for any domain $d'$ belonging to a monomer in $p$, if $loc(d') = (x',y')$,  $y' \le y_{max}(p)$.
Also, let $L(p)$ be $\{(x,y) | \exists\ m \in p \text{ and } d\in m \text{ with } loc(d) = (x,y) \}$. 
We now define a relation on the pairs in $L_{p}$. For $(x,y), (x',y')\in L_p$, $(x,y) <_{zz} (x',y')$ iff 1) $x < x'$, 2) $x=x'$, $x$ is even, and $y < y'$, or 3) $x=x'$, $x$ is odd, and $y> y'$.
In addition, for computation or end  monomers $m$ and $m'$, we write $m <_{zz} m'$ iff $m$ contains a domain $d$ and $m'$ contains a domain $d'$ such that $loc(d)<_{zz} loc(d')$.

\subsubsection{Monomer collection $\vvc$}

We now define a class of monomer counts which serve to satisfy the constraints necessary for our construction to correctly simulate $M$: in $\vvc$ it must be that the number of seed monomers $\#(\vvc,m_i) \le c_{min} \le k_{min}$.

\subsubsection{A stable configuration of $\vvc$}

We now consider the configuration of $\vvc$, which we call $\alpha$, that consists of exactly 4 types of polymers describe next.
\begin{enumerate}
    \item Saturated polymers consisting of a single seed monomer, $s\cdot 2t$ computation monomers corresponding to a correct and complete computation to $M$ on input $i$, and $s$ end monomers such that all domains are bound to a corresponding codomain. We call such a polymer a \emph{computation polymer}. \label{item:comp-polymers}
    \item Polymers that consist of a single capping monomer with two codomains bound to two domains of a single computation monomer. 
    \item Polymers that consist of a single capping monomer with two codomains bound to two domains of a single end monomer.  
    \item Polymers that are single capping monomers. 
\end{enumerate}
 We will  prove that $\alpha$ is stable, and then show that $\alpha$ is the only stable configuration of $\vvc$. We begin by showing that $\alpha$ is saturated.

\begin{claim}\label{clm:saturated}
    $\alpha$ is saturated and $H(\alpha) = 2\cdot \#comp(\vvc) + 2\cdot \#end(\vvc) + s$
\end{claim}

\begin{proof}(Proof of Claim~\ref{clm:saturated})
First, let $n = s\cdot 2t + s$ denote the number of monomers with type in $\calM_{comp, xy}\cup \calM_{end, xy}$ contained in a computation polymer of $\alpha$. Then to show that $\alpha$ is saturated, we define a sequence of configurations $\alpha_i$ for $i\in \N$ between $0$ and $n+2$, where $\alpha_{n+2} = \alpha$, $\alpha_0$ is saturated, and $\alpha_i$ is saturated if $\alpha_{i-1}$ is saturated.  We will describe $\alpha_1$ and define $\alpha_i$ recursively so that $\alpha_i$ is obtained from $\alpha_{i - 1}$ for $1\leq i \leq n$.  Then, we define $\alpha_{n+1}$ and $\alpha_{n+2}$ based on $\alpha_n$ and $\alpha_{n+1}$ (respectively) as a special case. 

\begin{figure}[ht]
    \centering
    \includegraphics[width=4.0in]{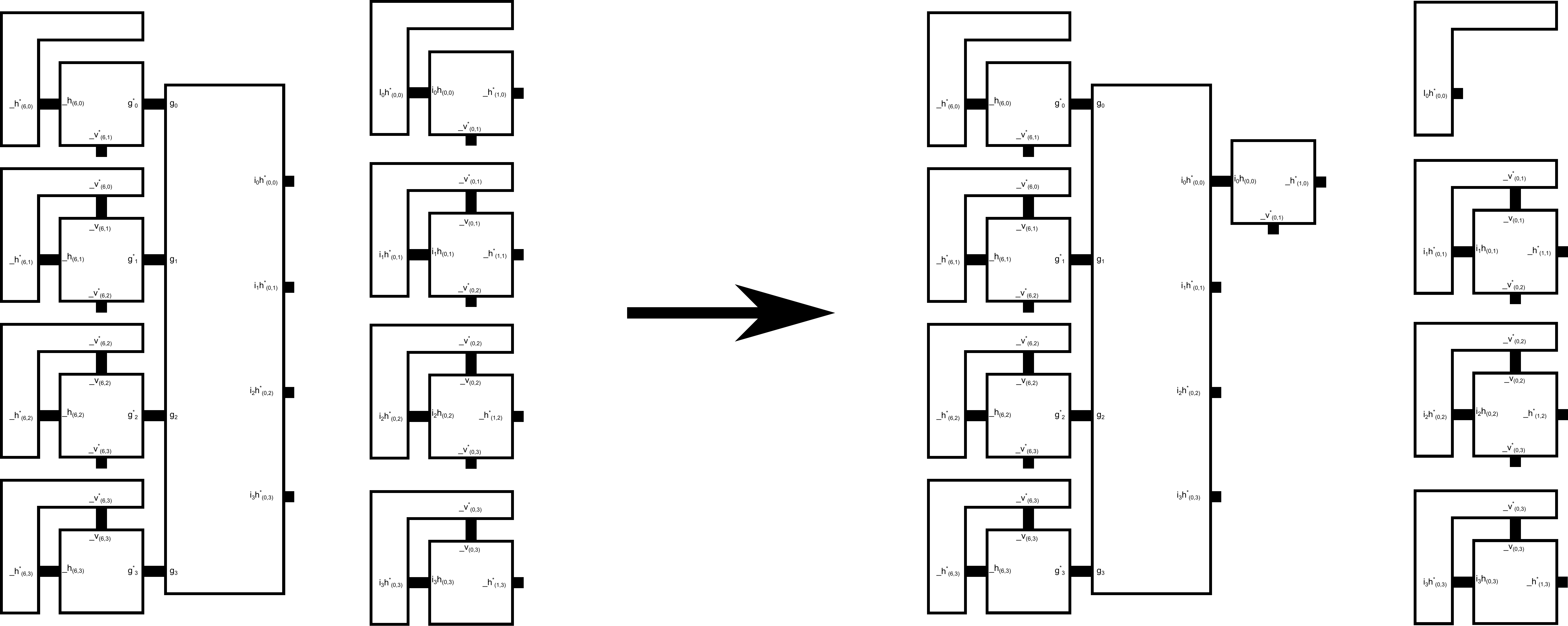}
    \caption{A schematic example of of $\alpha_{0}$ when $s = 4$ and $t=3$. Domain symbols corresponding to the glues of a tile set simulating the Turing machine have been suppressed and the symbol ``\_'' is shown. (Left) A depiction of a polymer containing a seed monomer of $\alpha_{0}$ and 4 polymers of $\alpha_{0}$ consisting of a single computation monomer and a single capping monomer. The polymer containing a seed monomer also contains 4 end monomers as shown, and the binding domains of each of these end monomers is bound to a codomain of a capping monomer. (Right) An example of of $\alpha_{1}$.}
    \label{fig:capping1}
\end{figure}

We now define $\alpha_0$. Shown in Figure~\ref{fig:capping1}, $\alpha_0$ is the configuration where each computation monomer and each end monomer each have two binding domains bound to two binding domains of a single capping monomer. In addition, all of the $s$ ``wrapping'' domains of the seed monomer are bound to a codomain of an end monomer. Note that $\alpha_0$ exists and is saturated as 1) the number of capping monomers matching a computation or end monomer is assumed to be greater than or equal to the said computation or end monomer, and 2) the number of end monomers with any type in $\calM_{end,xy}$ is greater than the number of seed monomers. We note that $H(\alpha_0) = 2\cdot \#comp(\vvc) + 2\cdot \#end(\vvc) + s$. This enthalpy count follows from the fact that each computation monomer or end monomer is bound to a single capping monomer and each seed monomer is bound to $s$ end monomers.

\begin{figure}[ht]
    \centering
    \includegraphics[width=4.0in]{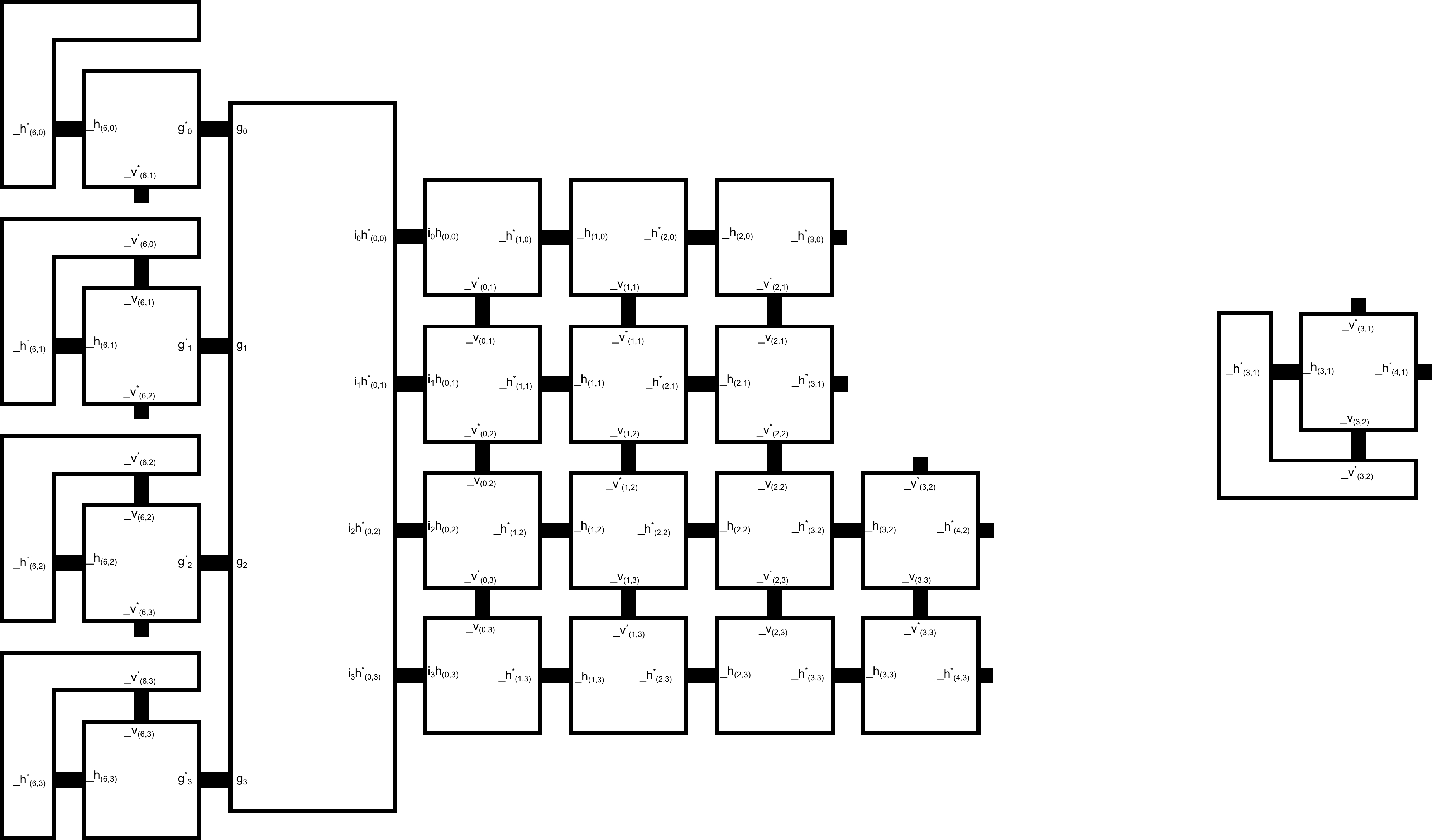}
    \caption{An exmaple depiction of $\alpha_{14}$ where $s = 4$ and $t=3$. Domain symbols corresponding to the glues of a tile set simulating the Turing machine have been suppressed and the symbol ``\_'' is shown. In $\alpha_{14}$, any monomer containing a seed monomer also contains $14$ computation monomers. A polymer consisting of a single computation monomer and a single capping monomer is also shown. Assuming that this computation monomer has appropriate domains, in order to obtain $\alpha_{15}$, the bonds between this computation monomer and capping monomer are broken, and two new bonds are formed between the domains of the computation monomer and codomains of monomers belonging to the polymer containing the seed monomer. Such ``bond swapping'' is performed for all such seed-monomer-containing polymers of $\alpha_{14}$ to obtain $\alpha_{15}$.}
    \label{fig:capping2}
\end{figure}

Now we recursively define $\alpha_i$. 
Let $i\in \N$ be such that $1\leq i \leq n$. We describe how to obtain $\alpha_i$ from $\alpha_{i-1}$. 
Enumerate the computation monomer and end monomer types of $\hat{p}$, $t_0, t_2, \dots, t_{n-s - 1}$ such that $t_{j-1} <_{zz} t_j$ for all $j$ such that $2\leq j\leq n$.

Let $d_1$ and $d_2$ be domains of monomers with types in $t_{i-1}$. 
Then, for each polymer $p$ in $\alpha_{i-1}$ that contains a seed monomer such that $p$ contains two monomers each with an unbound codomain, where one codomain is $d_1^*$ and the other is $d_2^*$, let $m$ be a computation monomer or end monomer in $\alpha_{i-1}$ with type $t_{i-1}$ such that the two domains of $m$ are bound to two codomains of a capping monomer, $K$ say. Moreover, if $t_{i-1}$ is in $\calM_{end,xy}$, choose $m$ so that $m$ belongs to $p$ if $m$ is a monomer in $p$, and otherwise choose $m$ such that $m$ belongs to a polymer consisting of a single capping monomer and a single end monomer. 
Next, we break the $d_1$ and $d_2$ bonds between $m$ and $K$, and then bind $d_1$ and $d_2$ of $m$ to $d_1^*$ and $d_2^*$ of the monomers belonging to the polymer $p$. Figure~\ref{fig:capping2} depicts this ``bond swapping''. 
$\alpha_i$ is defined to be this resulting configuration. 

Note that since $\alpha_{i}$ is obtained from $\alpha_{i-1}$ be breaking $2$ bonds and then creating $2$ bonds, $\alpha_i$ and $\alpha_{i-1}$ have the same enthalpy. I.e. $H(\alpha_i) = H(\alpha_{i-1})$. 

\begin{figure}[ht]
    \centering
    \includegraphics[width=4in]{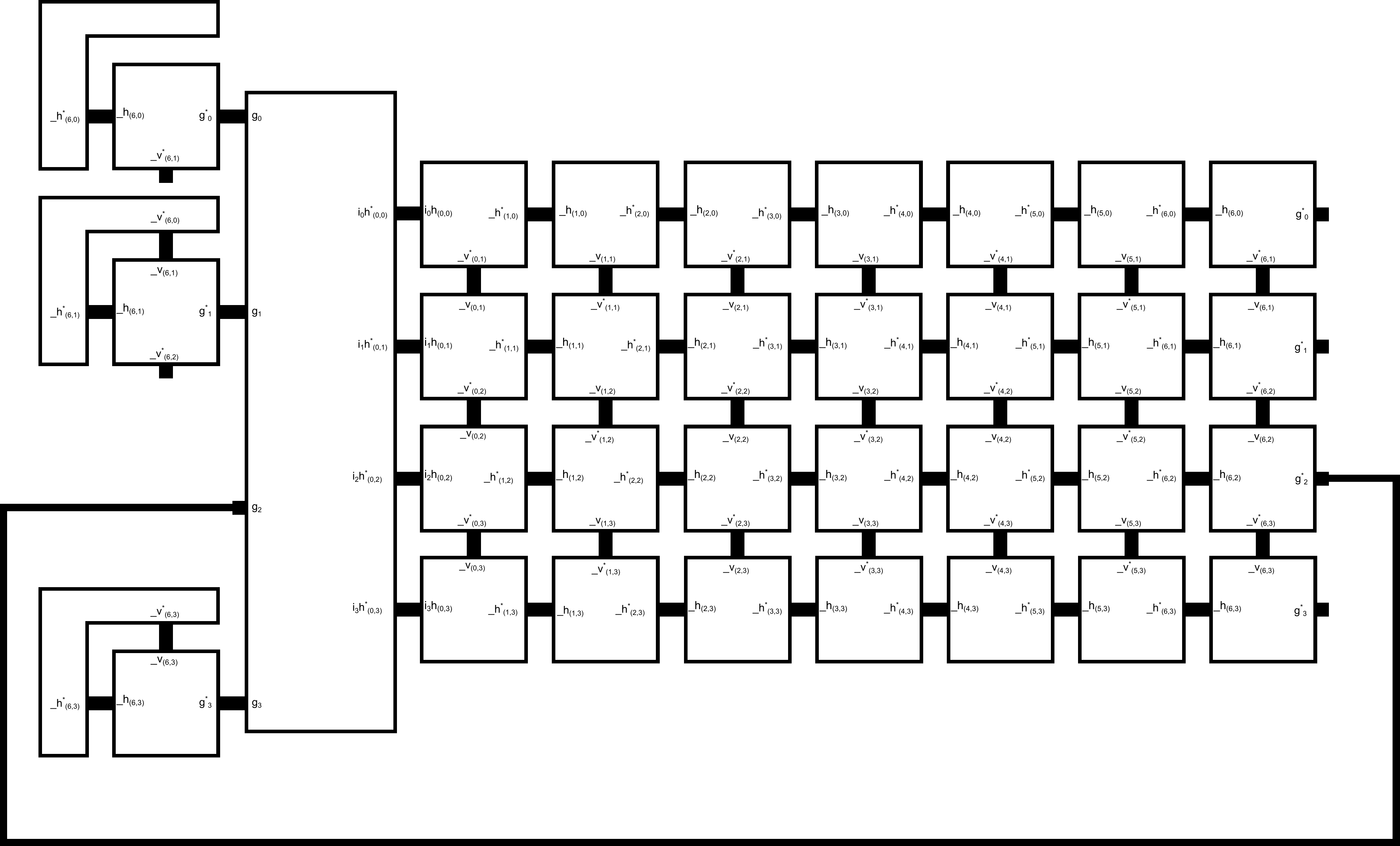}
    \caption{An exmaple of $\alpha_{n}$ where $s=4$ and $t=3$. Domain symbols corresponding to the glues of a tile set simulating the Turing machine have been suppressed and the symbol ``\_'' is shown. The rightmost column of monomers are end monomers. This example illustrates how some of the domains labeled $g_0$, $g_1$, $g_2$, and $g_3$ can be bound to end monomers that are bound to capping monomers ($g_0$, $g_1$, and $g_3$ in this example) and some of these domains ($g_2$ in particular) may be bound to end monomers bound to computation monomers.}
    \label{fig:capping3}
\end{figure}

Note that $\alpha_n$, defined recursively, has the following properties. For any polymer, $p$ say, in $\alpha_n$ containing a seed monomer, $p$ contains a single seed monomer, the sum of the number of computation monomers in $p$ and the number end monomers in $p$ is $n$, and the only capping monomers in $p$ have codomains that are bound to domains of end monomers in $p$ and these end monomers are bound to $s$ many domains labeled $g_j$ (for $0\leq j < s$) belonging to the seed monomer. Figure~\ref{fig:capping3} depicts $\alpha_n$.

\begin{figure}[ht]
    \centering
    \includegraphics[width=4in]{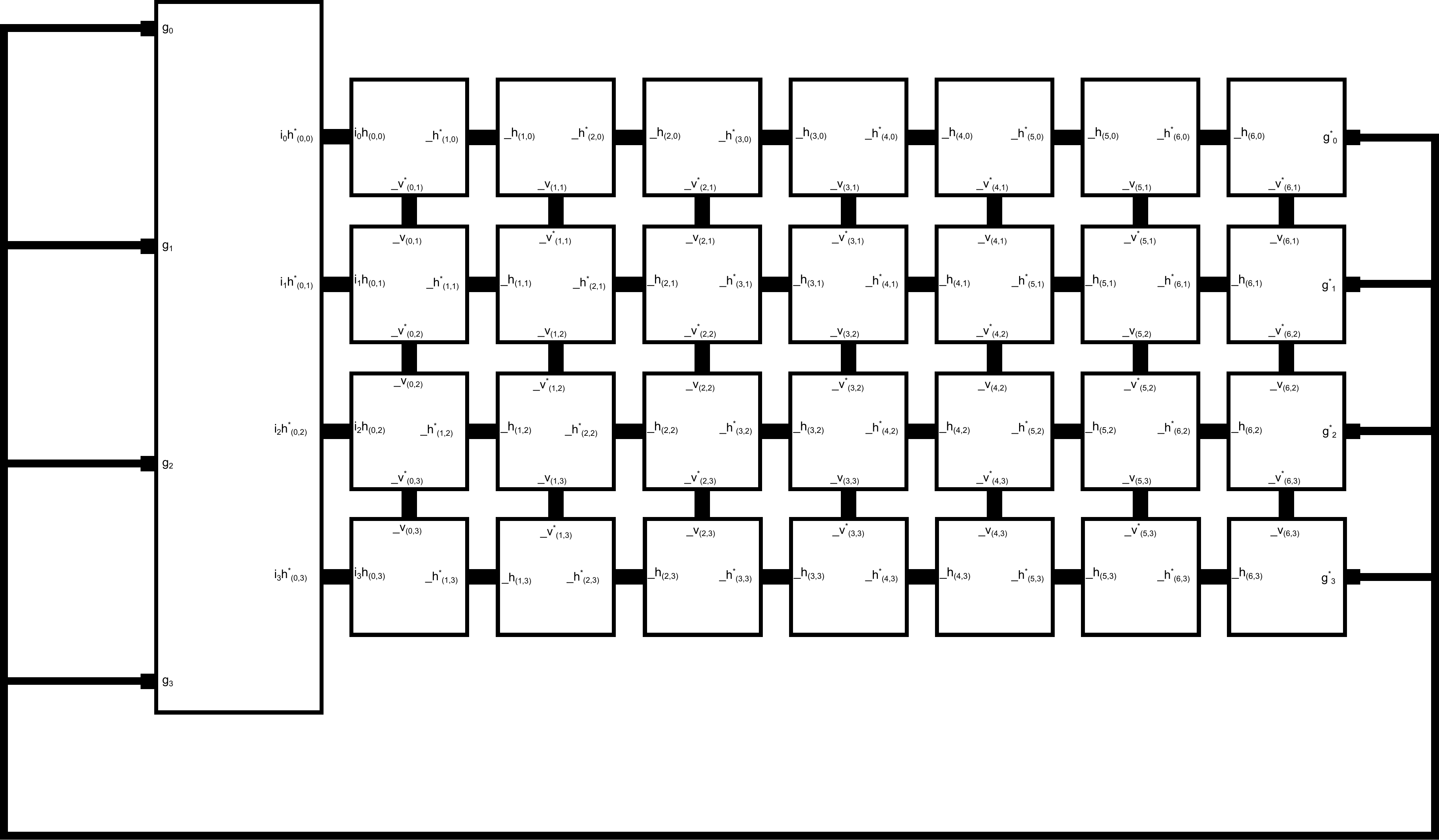}
    \caption{An example of $\alpha_{n+1}$ where $s=4$ and $t=3$. Domains symbols corresponding to the glues of a tile set simulating the Turing machine have been suppressed and the symbol ``\_'' is shown.}
    \label{fig:capping4}
\end{figure}

We then obtain $\alpha_{n+1}$ from $\alpha_n$ as follows. For all polymers $p$ containing a seed monomer and for each end monomer $e$ in $p$ with domains that are bound to two codomains of a capping monomer, let $j\in \N$ be such that $g_j$ is a domain of $e$. Note that by definition of $\alpha_{n}$, $p$ contains an end monomer $e'$ that is distinct from $e$ and contains a codomain $g_j^*$. The domains for $e'$ are bound to codomains of computation monomers in $p$. Next, break the bond between $g_j^*$, the domain in $e$, and $g_j^*$, a domain of the seed monomer of $p$, and then create a bond between the domain $g_j$ of the seed monomer of $p$ and $g^*$, the domain in $e'$. The resulting configuration is $\alpha_{n+1}$. Since every broken bond preceeds the creation of a bond, $H(\alpha_n) = H(\alpha_{n+1})$. 

Finally, we note that $\alpha_{n+1} = \alpha$ by definition of $\alpha_{n+1}$ and $\alpha$. Moreover, since $\alpha_1$ is saturated, and for $i\in \N$ such that $1\leq i \leq n+1$, $H(\alpha_i) = H(\alpha_{i-1})$, we see that $\alpha$ is saturated. Moreover, $H(\alpha) = 2\cdot \#comp(\vvc) + 2\cdot \#end(\vvc) + s$. 
\end{proof}

To show that $\alpha$ is stable we must show that $S(\alpha)$ is greater than or equal to the entropy for any saturated configuration. The following claim will help prove an upper bound on the entropy of any saturated configuration.

\begin{claim}\label{clm:bound-to-seed-or-cap}
    Let $\beta$ be an arbitrary saturated configuration of $\vvc$. Moreover, let $p$ be a polymer of $\beta$ such that $p$ contains a computation monomer or an end monomer $m$. Then, $p$ contains either a seed monomer or a capping monomer (or both).
\end{claim}

\begin{proof}(Proof of Claim~\ref{clm:bound-to-seed-or-cap})
Let $x \in \{0, 1, \dots, 2t-1 \}$ and $y\in \{0,1, \dots, s-1\}$, and let $m_{(x,y)}$ denote either a computation or end monomer such that $loc(m_{(x,y)}) = (x,y)$. We prove the claim by induction on $x$. 

For the base case, suppose that $p$ contains some computation monomer or end monomer $m_{(x,y)}$ for $x=0$ and some $y$. Let $dh_{(x,y)}$ and $dv_{(x,y)}$ be the domains of $m_{(x,y)}$. Note that as $\beta$ is saturated, $dh_{(x,y)}$ must be bound to a codomain $dh_{(x,y)}^*$ belonging to either a seed monomer or a capping monomer. This follows from the fact that there are no other monomer types with codomain $dh^*_{(x,y)}$. This proves the base case.

Now, for the sake of induction, assume that for some $x>0$, if $p$ contains a computation or end monomer, $m$ say, with $loc(m) = (x-1, y)$, then $p$ contains a capping monomer or a seed monomer. Then, suppose that $p$ contains some computation monomer or end monomer $m_{(x,y)}$ with $loc(m_{(x,y)}) = (x,y)$ for some $y$. Once again, let $dh_{(x,y)}$ and $dv_{(x,y)}$ be the domains of $m_{(x,y)}$. Then, since there exists a capping monomer with codomains $dh^*_{(x,y)}$ and $dv^*_{(x,y)}$, the domain $dh_{(x,y)}$ is either bound to a capping monomer, or bound to a codomain belonging to a computation monomer.  In the former situation, the claim holds. In the latter situation, there is a computation monomer, which we denote by $m_{(x-1, y)}$ such that $loc(m_{(x-1, y)}) = (x-1, y)$. By the inductive hypothesis, the polymer containing $m_{(x-1, y)}$ contains either a capping monomer or a seed monomer. Therefore, $p$ contains either a seed monomer or a capping monomer. Therefore, the claim holds.
\end{proof}

The following claim gives an upper bound on the entropy of any saturated configuration.

\begin{claim}\label{clm:entropy-bound}
    Let $\beta$ be a saturated configuration of $\vvc$. $S(\beta) \leq \#(\vvc, m_i) + \#cap(\vvc)$.
\end{claim}

\begin{proof}(Proof of Claim~\ref{clm:entropy-bound})
Since $\calM = \calM_{comp,xy} \cup \calM_{cap,xy} \cup \calM_{seed} \cup \calM_{end,xy}$, the claim follows from Claim~\ref{clm:bound-to-seed-or-cap}.

\end{proof}

We now note that $\alpha$ is a stable configuration of $\vvc$.

\begin{claim}\label{clm:alpha-stable}
    $\alpha\in [\vvc]_{\Box}$.
\end{claim}

\begin{proof}(Proof of Claim~\ref{clm:alpha-stable})
We note that $S(\alpha) = \#(\vvc, m_i) + \#cap(\vvc)$. Therefore, the claim follows from Claim~\ref{clm:saturated} and Claim~\ref{clm:entropy-bound}.

\end{proof}

\begin{claim}\label{clm:stable-seed-or-cap}
    Suppose that $\beta \in [\vvc]_{\Box}$. Then, for all polymers $p$ contained in $\beta$, $p$ contains either a single capping monomer or a single seed monomer, but not both monomer types.
\end{claim}

\begin{proof}(Proof of Claim~\ref{clm:stable-seed-or-cap})
This claim follows from Claim~\ref{clm:bound-to-seed-or-cap} and the fact that $S(\alpha) = \#(\vvc, m_i) + \#cap(\vvc) = S(\beta)$. 
\end{proof}

\begin{claim}\label{clm:2-comps}
    Suppose that $\beta \in [\vvc]_{\Box}$. Then, for all polymers $p$ contained in $\beta$, if $p$ contains two monomers with type in $\calM_{comp,xy}\cup \calM_{end,xy}$, then $p$ contains a seed monomer.
\end{claim}

\begin{proof}(Proof of Claim~\ref{clm:2-comps})
We will prove this claim by contradicting Claim~\ref{clm:stable-seed-or-cap}. Suppose that $p$ is a polymer of $\beta$ such that $p$ contains two distinct monomers with type in $\calM_{comp,xy}\cup \calM_{end,xy}$. For the sake of contradiction, assume that $p$ does not contain a seed monomer. 

Let $D\subseteq \calD$ be the set of all domains belonging to either a computation monomer or an end monomer contained in $p$. Let $L\subset \Z^2$ be the set $\{(x,y) | \exists\ d\in D \text{ with } loc(d) = (x,y)\}$. Then let $L_{xmin}$ be $\{(x,y) | \forall\ (x', y')\in L, x \leq x'\}$. Observe that $L_{xmin}\neq \emptyset$ as $p$ contains at least two monomers with type in $\calM_{comp,xy}\cup \calM_{end}$. 

Let $m_1$ be an arbitrary monomer in $L_{xmin}$, let $(x_1, y_1) \in \Z^2$ be such that $m_1$ contains a domain, $d_1h$ say, with $loc(d_1h) = (x_1, y_1)$ and $or(d_1h) = h$. Let $a_1v$ denote the domain belonging to $m_1$ such that $or(a_1v) = v$, and let $b_1h^*$ and $c_1v^*$ denote the codomains belonging to $m_1$ such that $or(b_1h^*) = h$ and $or(c_1v^*) = v$.  

Note that $d_1h$ must be bound to a codomain belonging to a capping monomer. This follows from the following observations. 
1) As $\beta$ is saturated, Claim~\ref{clm:saturated} implies that $H(\beta) = 2\cdot \#comp(\vvc) + 2\cdot \#end(\vvc) + s$, and 
2) $d_1h$ is contained in $L_{xmin}$.
By (1), $d_1h$ must be bound to some codomain belonging to some monomer. 
By (2), $d_1h$ cannot be bound to a codomain belonging to a monomer with type in $\calM_{comp,xy}\cup \calM_{end}$, as such a monomer would imply that $d_1h\not \in L_{xmin}$. Then, as $\beta$ does not contain a seed monomer by assumption, it must be the case that $d_1h$ is bound to some codomain belonging to a capping monomer. Let $c_1$ denote this monomer.

By assumption, $p$ contains a second monomer with type in $\calM_{comp,xy}\cup \calM_{end}$. Hence, $m_1$ must contain a domain or codomain that is bound to a codomain or domain of a  monomer, which we denote by $m_2$, with type in $\calM_{comp,xy}\cup \calM_{end}$. We consider the following situations. 
\begin{enumerate}
    \item $a_1v$ is bound to a codomain, $a_1v^*$, of $m_2$,\label{item:a-bound}
    \item $c_1v^*$ is bound to a domain, $c_1v$, of $m_2$, or\label{item:c-bound}
    \item (1) and (2) do not hold, and $b_1h^*$ is bound to a domain, $b_1h$ of $m_2$.\label{item:b-bound}
\end{enumerate}

If (\ref{item:a-bound}) holds, then there exists a domain $d_2h$ belonging to $m_2$ and an integer $y_2$ such that $loc(d_2h) = (x, y_2)$ and $or(d_2h)=h$. 
Recall the argument above showing that $d_1h$ must be bound to a codomain of a capping monomer. By a similar argument, if (\ref{item:a-bound}) holds, $d_2h$ must be bound to a codomain of a capping monomer, $c_2$ say.  Moreover, $c_2 \neq c_1$ as $c_1$ contains a codomain $d_1h^*$ and $c_2$ contains a codomain $d_2h^*$ and there does not exist a capping monomer type containing the two codomains $d_1h^*$ and $d_2h^*$. Thus, by Claim~\ref{clm:stable-seed-or-cap}, we have arrived at a contradiction when (\ref{item:a-bound}) holds. Using an analogous argument, when (\ref{item:c-bound}) holds, we arrive at the same contradiction.

Finally, suppose that (\ref{item:b-bound}) holds. We note that $loc(b_1h) = (x+1, y)$. 
Let $a_2v$ denote the domain belonging to $m_2$ such that $or(a_2v) = v$, and let $b_2h^*$ and $c_2v^*$ denote the codomains belonging to $m_2$ such that $or(b_2h^*) = h$ and $or(c_2v^*) = v$. Then either 
A) $a_2v$ is unbound,
B) $a_2v$ is bound to a domain belonging to a capping monomer, or
C) $a_2v$ is bound to a domain belonging to a monomer with type in $\calM_{comp,xy}\cup \calM_{end}$.

If (A) holds, then $\beta$ is not saturated as the enthalpy of a saturated configuration must be $2\cdot \#comp(\vvc) + 2\cdot \#end(\vvc) + s$ by Claim~\ref{clm:saturated}. Thus, (A) cannot hold.

Now suppose that (B) holds. Then, let $c_3$ denote the monomer with codomain $a_2v^*$ bound to the domain $a_2v$. Note that $c_1 \neq c_3$. This follows from the fact that for the two codomains of $c_1$, $loc$ maps these domains to $(x,y)$, yet for the two codomains of $c_3$, $loc$ maps these domains to $(x+1, y')$ for some $y'\in \N$. Since $c_1 \neq c_3$, we arrive at a contradiction of Claim~\ref{clm:stable-seed-or-cap} when (B) holds.

Finally, suppose that (C) holds. Then, let $m_3$ denote the monomer with type in $\calM_{comp,xy}\cup \calM_{end}$ such that $m_3$ has a codomain $a_2v^*$ bound to $a_2v$ of $m_2$.  Now, there exists a domain $d_3h$ belonging to $m_3$ such that $or(d_3h) = h$ and $loc(d_3h) = (x+1, y')$ for $y' \in \N$. Note that $y'$ is either $y+1$ or $y-1$. In either case, $d_3h$ is either bound to a codomain of a capping monomer, $c_4$ say, or bound to a codomain of a computation monomer, $m_4$ say. In the former case, we note that $c_4 \neq c_1$. This follows from the fact that for the two codomains of $c_1$, $loc$ maps these domains to $(x,y)$, yet for the two codomains of $c_4$, $loc$ maps these domains to $(x, y')$ for $y'$ either $y+1$ or $y-1$. Since $c_1 \neq c_4$, we arrive at a contradiction of Claim~\ref{clm:stable-seed-or-cap}. In the latter case, we note that $m_4$ must contain a domain $d_4h$ such that $or(d_4h) = h$ and $loc(d_4h) = (x, y')$. Now, as $d_4h\in L_{xmin}$ and $\beta$ is saturated, it must be the case that $d_4h$ is bound to a codomain of a capping monomer, $c_5$. We finally note that $c_1\neq c_5$ since $loc$ maps the codomains belonging to $c_1$ to $(x,y)$ and maps the codomains belonging to $c_5$ to $(x,y')$ and $y\neq y'$. Once again, we arrive at a contradiction of Claim~\ref{clm:stable-seed-or-cap} when under the assumption that (C) holds. In any case, if (\ref{item:b-bound}) holds, we arrive at a contradiction.  

Thus, in any case, (\ref{item:a-bound}), (\ref{item:c-bound}), or (\ref{item:b-bound}), we arrive at a contradiction. Hence, there is no polymer $p$ of $\beta$ such that $p$ contains two distinct monomers with type in $\calM_{comp,xy}\cup \calM_{end,xy}$ and does not contain a seed monomer, which was what we wanted.
\end{proof}

\begin{claim}\label{clm:stable-grid}
Suppose that $\beta \in [\vvc]_{\Box}$. Let $p$ be a polymer contained in $\beta$ such that $p$ contains a seed monomer. Then, $L(p) = \{(x,y) | 0\leq x\leq 2t-1 \text{ and } 0\leq y \leq s - 1\}$.
\end{claim}

\begin{proof}(Proof of Claim~\ref{clm:stable-grid})
Let $p$ be a polymer contained in $\beta$ such that $p$ contains a seed monomer. Fix $y\in \N$ such that $0\leq y \leq s-1$. Let $g_y$ be the domain belonging to the seed monomer of $p$ such that there is an end monomer with codomain $g_y^*$. Note that as $\beta$ is saturated, $g_y$ must be bound to some codomain of an end monomer. Let $m_{(2t-1,y)}$ denote this end monomer and note that $loc(m_{(2t-1,y)}) = (2t-1,y)$. Hence, $(2t-1, y)\in L(p)$. Then, $m_{(2t-1,y)}$ contains a domain $dh_{(2t-1,y)}$ such that $or(dh_{(2t-1,y)}) = h$ and $loc(dh_{(2t-1,y)}) = (2t-1,y)$. Once again, as $\beta$ is saturated, $dh_{(2t-1,y)}$ must be bound to a codomain of a monomer which we denote by $m_{(2t-2,y)}$. By Claim~\ref{clm:stable-seed-or-cap}, $m_{(2t-2,y)}$ cannot be a capping monomer as $p$ contains a seed monomer. Hence, $m_{(2t-2,y)}$ is a computation monomer. We note that $loc(m_{(2t-2,y)}) = (2t-1, y)$, and hence, $(2t-2, y)\in L(p)$. Repeating this argument $2t-2$ more times, we see that $(x, y)\in L(p)$ for all $x$ such that $0\leq x\leq 2t-1$ and this fixed $y$ value. Moreover, as $y$ is arbitrary, we see that $L(p) = \{(x,y) | 0\leq x\leq 2t-1 \text{ and } 0\leq y \leq s - 1\}$.
\end{proof}

\begin{claim}\label{clm:seed-comp}
    Let $\hat{p}$ denote a polymer in $\alpha$ that contains a seed monomer.
    Suppose that $\beta \in [\vvc]_{\Box}$. Then, for all polymers $p$ contained in $\beta$, if $p$ contains a seed monomer, then the set of monomer types for monomers in $p$ is exactly the set of monomer types for monomers in $\hat{p}$.
\end{claim}

\begin{proof}(Proof of Claim~\ref{clm:seed-comp})
By Claim~\ref{clm:stable-grid}, $L(p) = \{(x,y) | 0\leq x\leq 2t-1 \text{ and } 0\leq y \leq s - 1\}$.  For a polymer $p$, let $T(p)$ denote the set of monomer types in $p$. For the sake of contradiction, assume that $T(p)\neq T(\hat{p})$. One can observer that there exists a type in $T(p)$ that is not in $T(\hat{p})$.
To see this, note that by Claim~\ref{clm:stable-grid}, $T(p)\geq T(\hat{p})$, so it cannot be the case that $T(p)\subsetneq T(\hat{p})$.

 Then, there exists a computation or end monomer in $p$ with type in $T(p) \setminus T(\hat{p})$. Let $m$ be a computation or end monomer in $p$ such that 1) the type of $m$ is in $T(p) \setminus T(\hat{p})$, and 2) for any computation or end monomer $m'$ with type in $T(p) \setminus T(\hat{p})$, $loc(m) <_{zz} loc(m')$. (Recall the definition of the relation $<_{zz}$ defined in Section~\ref{sec:hard-coded-notation}.) Consider the domains of $m$. Let $d_h$ and $d_v$ denote the two domains of $m$ such that $or(d_h)=h$ and $or(d_v)=v$.  Note that it cannot be the case that there are monomers in $\hat{p}$ with codomains $d_h^*$ and $d_v^*$ for otherwise $m$ would have a monomer type in $T(\hat{p})$. So, either $\hat{p}$ does not contain a monomer with codomain $d_h^*$ or does not contain a monomer with codomain $d_v^*$.
First suppose that $d_h^*$ is not contain in any monomer belonging to $\hat{p}$. Then, either the domain $d_h$ of $m$ (in $\beta$) is not bound to any codomain, in which case $\beta$ is not saturated, or $d_h$ is bound to a capping monomer or seed monomer, which contradicts Claim~\ref{clm:stable-seed-or-cap}. In either case, we arrive at a contradiction. Now suppose that $\hat{p}$ does not contain a monomer with codomain $d_v^*$. An identical argument to the one for $d_h^*$ leads to the same contradiction. Therefore, $T(p) = T(\hat{p})$, which proves the claim.
\end{proof}

The following claim is our final claim before we finish the proof of Theorem~\ref{thm:hard-coded-sim}.

\begin{claim}\label{clm:one-stable}
    $|[\vvc]_{\Box}| = 1$.
\end{claim}

\begin{proof}(Proof of Claim~\ref{clm:one-stable})
This claim follows from Claims~\ref{clm:stable-seed-or-cap},~\ref{clm:2-comps}, and~\ref{clm:seed-comp}.
\end{proof}

We now finish the proof of Theorem~\ref{thm:hard-coded-sim}.
Referring to Definition~\ref{def:simulate-TM}, the set input monomer types, $\mathcal{I}$, is $\calM_{seed}$. The set of output monomer types $\mathcal{O} = \calM_{end}$. The encodings $E_{input}$ and $E_{output}$ are given by the construction. 
We then note that by Claim~\ref{clm:one-stable}, there is one stable configuration,  $\alpha$ say, in $[\vvc]_\Box$. Let $p$ be any polymer in $\alpha$ such that $p$ contains a monomer with type in $\calM_{seed}$. Then, $p$ is equal to $\hat{p}$ ($\hat{p}$ is defined above). Therefore, $p$ contains a monomer with type in $\mathcal{I}$, namely the seed monomer, such that
\begin{enumerate}
\item by Claim~\ref{clm:stable-seed-or-cap}, $p$ contains only one monomer $m$ with type in $\mathcal{I}$, and moreover, $E_{input}(m) = i$ by construction of types $\calM_{seed}$, and
\item the subset of monomers in $p$ consisting of exactly the monomers with type in $\mathcal{O}$, which we denote by $p_{\mathcal{O}}$, is such that $E_{output}(p_{\mathcal{O}}) = M(i)$, where $M(i)$ is the output of the Turing machine $M$ on input $i$. This follows from the description of $p = \hat{p}$.
\end{enumerate}
\end{proof}

\section{Details of arbitrary Boolean circuit simulation via TBN}\label{sec:circuits-append}
\subsection{Kinetic pathway and geometric plausibility}
Abandoning tile assembly gives no argument of geometric plausibility, but the construction given here satisfies that the bond graph of the polymer containing the seed in the stable configuration of the TBN is mostly isomorphic to the directed, acyclic graph (DAG) which models the circuit--- mostly isomorphic in the sense that it includes just $n$ more edges than the DAG for the circuit, where $n$ is the input length, and instead of $n$ input nodes, the TBN has one input monomer bound to the first set of gate monomers; the inter-connectivity of the gate monomers.

Intuitively, this gives that the ``geometric complexity'' of the circuit being simulated determines the ``geometric complexity'' of the polymer which performs the computation--- e.g., if your circuit is planar, so is the bond graph of the computing polymer in the stable configuration of the TBN.
\subsection{Arbitrary Boolean circuit simulation construction}
An $n$-input, $m$-output, $k$-fan-in, $l$-fan-out \emph{arbitrary circuit} is modeled by a directed, acyclic graph $C$ with three types of nodes: $n$ \emph{input} nodes with in-degree zero, $m$ output nodes with in-degree one and out-degree zero, and \emph{gate} nodes with maximum in-degree $k$ and out-degree $l$.
Each input node is associated with a unique Boolean variable selected from an \emph{input} $\{b_1, b_2, \cdots, b_n\} $ with each $b_j \in \{0,1\}$.
Each gate node $g_k$ with in-degree $i$ and out-degree $o$ is associated with a function $f_k : \{0,1\}^i \rightarrow \{0,1\}^o$.
Given an input string for the input nodes, each outgoing edge from a gate node $g_k$ is associated with a Boolean variable, uniquely selected from $f_k(\iota)$ with $\iota$ a sequence of $\{0,1\}$'s determined by the incoming edges for the gate--- this is computed starting at the gate nodes which have incoming edges from the input nodes.
The output corresponding to the input $i$ is a sequence generated from the Boolean variables associated with the incoming edges to the output nodes (with an arbitrary ordering fixed for all inputs; e.g., the choice of selecting the output bit string from the output nodes must be consistent independent of the input bit string).

Using the notation and a similar definition as given in Section~\ref{sec:hard-coded-comp} of a TBN simulating a TM, the definition of a TBN simulating an arbitrary circuit follows:
\begin{definition}
A monomer collection $\vvc$ for the TBN $\calT$ \emph{simulates} an arbitrary circuit $C$ on input $I$ if and only if
there exist encodings $E_{input}$ and $E_{output}$ such that for every stable configuration $\alpha$ in $[\vvc]_\Box$, if $\alpha$ contains a polymer that contains a monomer with type in $\mathcal{I}$, then 
\begin{enumerate}
\item for $S_{in}$ the set of monomers in $\alpha$ with types in $\mathcal{I}$,  $E_{input}(S_{in}) = i$,
\item for $S_{out}$ the set of monomers in $\alpha$ with types in $\mathcal{O}$,  $E_{output}(S_{out}) = C(i)$, where $C(i)$ is the output of the Boolean circuit evaluated on input $i$.
\end{enumerate}
\end{definition}

Given an arbitrary circuit $C$ and an input $i$, we construct a simulating TBN $\calT = \{\calD, \calM \}$ by the following scheme, using the same decomposition as Section~\ref{sec:construction-comps} of monomers into four sets:
\begin{enumerate}
    \item The \emph{seed monomers}, constructed equivalently as described in Section~\ref{sec:construction-comps}: one seed monomer for each possible bit string input.
    \item The \emph{computation monomers}: for each node $v$ in the circuit of out-degree $o$, for all sequences of Boolean variables $s$ of length $i$ equal to the in-degree of $v$, with $f_v(s)_j$ equal to the $j^{th}$ bit of $f_v(s)$, let $d(s_k)$ be a domain constructed by concatenating $s_k$ with a subscript $u,v$ where $u$ is the node connected to $v$ by the edge which is associated with the input variable $s_k$--- e.g., the domain may be $0_{u,v}$ or $1_{u,v}$--- and let $d(f_v(s)_j)$ be a domain constructed by evaluating $f_v(s)_j$ and adding superscript $*$ and subscript $v,w$ where $w$ is the node attached to $v$ by the edge associated with $f_v(s)_j$--- e.g., the domain may be $0_{v,w}^*$ or $1_{v,w}^*$ depending on the input $s$ and $f_v$. Construct a monomer for each $v$ and $s$ as follows: $m_v^s = \{d(s_1), d(s_2), \cdots, d(s_i), d(f_v(s)_1), d(f_v(s)_2),\cdots d(f_v(s)_o)\}$.
    \item The \emph{end monomers}: for each output node $u$ in the circuit, include in $\calM$ two monomers: $\{0_{v,u}, g_{u}^*\}$ and $\{1_{v,u}, g_u^*\}$ where $v$ is the gate node sharing an edge with $u$. $g_u^*$ are the domains complementary to those on the seed monomers.
    \item The \emph{capping monomers}, constructed equivalently as described in Section~\ref{sec:construction-comps}; i.e., a capping monomer for each monomer in the computation and end monomer sets.
\end{enumerate}

Given this scheme, the proof of Theorem~\ref{thm:hard-coded-sim} follows for this construction as well, to prove the following corollary:

\begin{corollary}\label{cor:circuit_sim}
For any arbitrary circuit $C$, there exists a set of primary domain types $\calD$, and sets of monomer types $\calM$, $\calM_{seed}$, and $\mathcal{O} \subset \calM$ consisting of monomers with binding domains in $\calD\cup \calD^*$
such that, for any valid input $i$ to $C$, the following properties hold.

\begin{enumerate}
    \item there exists a monomer type $m_i$ in $\calM_{seed}$ such that $m_i$ encodes $i$, 
    \item for $\calM_i = \calM \cup \{m_i\}$, there exists a monomer collection $\vvc$ for TBN $\calT_i = (\calD, \calM_i)$ such that $\vvc$ simulates $C$ on input $i$, and
    \item the set of output monomer types for the simulation is equal to $\mathcal{O}$.
\end{enumerate} 
\end{corollary}

\section{Technical details of a negative result on porting computing systems from the aTAM to the TBN model}
\label{sec:aTAM-fail-details}%

In this section we provide the technical details of the argument of Section~\ref{sec:aTAM-fail}, which shows what occurs when standard aTAM systems which are designed to simulate Turing machines are treated as TBNs. This is similar to the TBN designed for the proof of Theorem~\ref{thm:hard-coded-sim} but without creating unique, hard-coded monomer types for each location.  Essentially, in a naive but more monomer-type-efficient construction, there is no guarantee of correct simulation, and we present a property of the simulated computation which can be inspected to determine if a sufficient (although perhaps not necessary) condition exists whereby the TBN can incorrectly simulate the computation.

\begin{figure}
    \centering
    \includegraphics[height=1.5in]{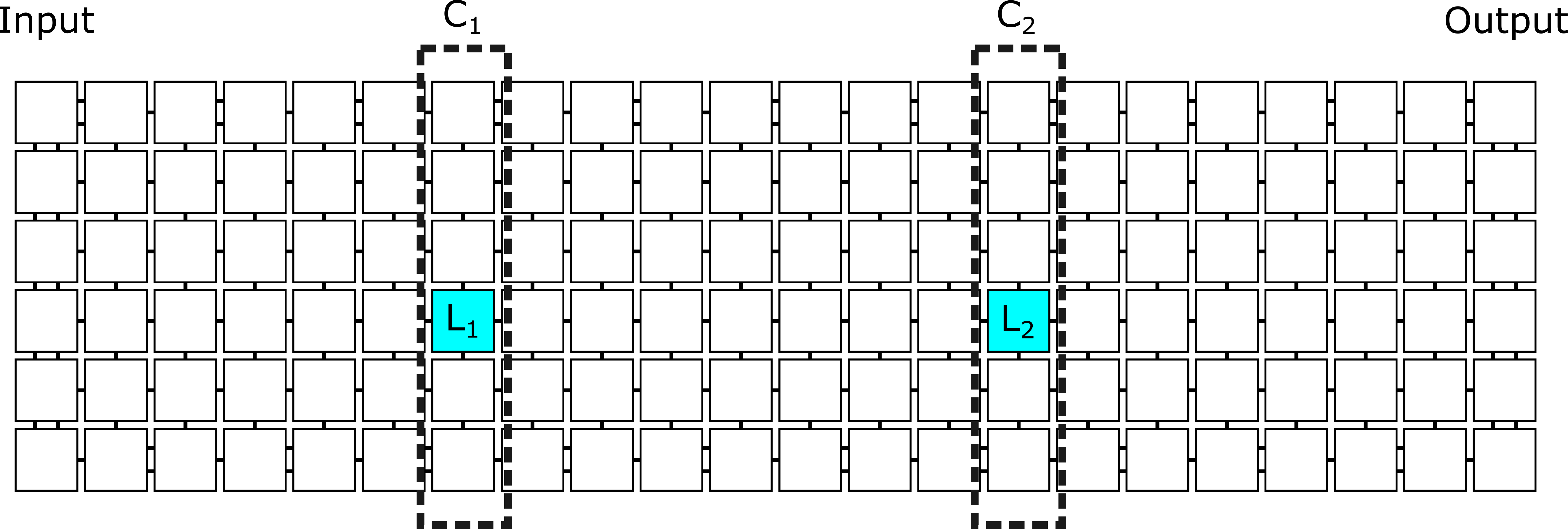}
    \caption{Schematic view of the terminal assembly of a zig-zag aTAM system which simulates a Turing machine.}
    \label{fig:comp-error-locations}
\end{figure}

Let $M$ be an $s$-space-bounded Turing machine and $T_M$ be a standard zig-zag aTAM tile set which simulates $M$.  (Note that this argument will also apply when $M$ is not space-bounded.) For $n \in N$ where $n$ is a valid input to $M$, let $T_n$ be the set of ``input'' tile types which assemble the binary representation of $n$ as a vertical column to serve as the input to $T_M$.  Then the aTAM system $\calT_{M(n)} = (T_M \cup T_n, \sigma_n, 2)$, where $\sigma_n$ is simply the first tile of $T_n$ at the origin, simulates $M(n)$. We refer to the unique terminal assembly of $\calT_{M(n)}$ as $S_n$ (i.e. simulation $n$).  Additionally, if $X$ is a set of coordinate locations, by $S_n(X)$ we refer to the subassembly of $S_n$ contained at the locations of $X$. Let $i \ne j \ne k$ be valid inputs to $M$, and $S_i$, $S_j$ and $S_k$ be the terminal assemblies of $\calT_{M(i)}$, $\calT_{M(j)}$, and $\calT_{M(k)}$, respectively, such that the following conditions hold:

\begin{enumerate}
  \setlength{\itemsep}{0pt}
  \setlength{\parskip}{0pt}
    \item The outputs $M(i) \ne M(k)$
    \item There exist columns (i.e. sets of all tile locations in a given column) $c_1$ and $c_2$, and individual tile locations $l_1$ and $l_2$ within $c_1$ and $c_2$, respectively, such that:
    \begin{enumerate}\label{cond:matching-locs}
        \item $S_i(c_1) = S_j(c_1)$ (i.e. both columns have the exact same tile types in the same locations) except at location $l_1$, where they have differing tile types with different glues on their west sides (which would represent different cell values for the respective simulated tape cells of $M$)
        \item $S_j(c_2) = S_k(c_2)$ except at location $l_2$ where they have tile types which differ in their west glues
        \item $S_i(l_1) = S_k(l_2)$
        \item $S_j(l_1) = S_j(l_2)$
        \item $S_i(c_2) \ne S_k(c_2)$
    \end{enumerate}
\end{enumerate}

Figure~\ref{fig:comp-error-locations} shows a schematic diagram of an example $S_n$ and locations $c_1$, $c_2$, $l_1$, and $l_2$.  The conditions of \ref{cond:matching-locs} essentially establish that although the various computations differ in their inputs and therefore at least portions of their computations, there are portions which are similar, which would be the case for many computations on nearly similar inputs.  The positions and types of these similarities are carefully chosen to support the argument which follows.

We'll define a TBN $\mathcal{N}$ which simulates $M$ using the same procedure as in Section~\ref{sec:construction-comps} (and thus skipping the blow-up performed in Section~\ref{sec:construction-explosion}).

Now, using Definition\ref{def:simulate-TM}, of simulation, and the same notions of inputs and outputs and encodings, we create a monomer collection $\vvc$ for $\mathcal{N}$ which has (only) the following constraints. Let $\#(\vvc,m)$ be the count of the monomers of type $m$ in $\vvc$, and  $c_{min} = min\{\#(\vvc,m) | m \in \mathcal{M}_{comp}\}$ (i.e. the number of copies of the computation monomer with the fewest number of copies), and $k_{min} = min\{\#(\vvc,m) | m \in \mathcal{M}_{cap}\}$. Then in $\vvc$ it must be that the number of seed monomers $\#(\vvc,m_i) \le c_{min} \le k_{min}$.

We now analyze $\mathcal{N}$ and note that while it is largely similar to the system $\calT$ used in the proof of Theorem~\ref{thm:hard-coded-sim}, with its domains and monomers designed so that it is possible for them to combine in a pattern analogous to the way the tiles in the zig-zag aTAM system $\calT_M$ combine to simulate $M$, there is an important difference in the two TBNs.  While in $\calT$ a set of unique monomer types exists for every row and column of the simulated computation and thus each monomer type occurs no more than one time in the polymer representing the computation, in $\mathcal{N}$ the same monomer types may be reused many times.  More specifically, if the conditions on $M$ and $\calT_M$ discussed above hold, it is possible to create a polymer which is saturated and nearly identical in structure to the polymer representing the correct computation, but which represents an invalid computation and output.  To prove this, we now describe how to construct such a polymer.

We know by the definition of the domains and polymers of $\mathcal{N}_n$ that it is possible for configurations to exist with polymers $p_i$, $p_j$, and $p_k$ (whose monomers have direct mappings with the tiles of assemblies $S_i$, $S_j$, and $S_k$) within the monomer collections $\vvc_i$, $\vvc_j$, and $\vvc_k$ of TBNs $\mathcal{N}_i$, $\mathcal{N}_j$, and $\mathcal{N}_k$, respectively.  We first note that such polymers would correctly represent simulations of their respective computations, and that in each polymer every domain of every included monomer could be bound.  Given the monomer counts of the collections, each input monomer could be contained within such a polymer, and all leftover computing monomers could be bound to their matching caps, and only caps in excess of the monomers that they ``cap'' remaining as unbound monomers, causing the configurations to be saturated. This results in configurations with maximum enthalpy and whose entropy is maximized among configurations with that maximum enthalpy (following the same argument as in the proof of the construction for Theorem~\ref{thm:hard-coded-sim}). However, let us now focus on the monomer collection utilizing TBN $\mathcal{N}_i$, which is simulating the computation $M(i)$, and assume that the conditions related to inputs $i$, $j$, $k$ and the aTAM assemblies which simulate them are true.  In this case, (and please refer to Figure~\ref{fig:comp-splicing-rewiring} for reference), we can form a polymer by including the input monomer $m_i$ and attaching all of the monomers which would have attached in $p_i$ up to, but not including, those which would logically represent column $c_1$ of $S_i$. We can then continue by attaching the monomers consistent with $S_j$ starting from column $c_1$ until, but not including, $c_2$.  Finally, we could finish the polymer by attaching the monomers consistent with $S_k$ from $c_2$ through the output.  Due to the constraints above (included in \ref{cond:matching-locs}), this could result in all monomer domains being bound to their logical neighbors except for $4$ domains.  These would be the domains corresponding the the western glues of the monomers representing tiles that are at locations $l_1$ and $l_2$, and those of the eastern glues of the monomers representing the tiles to their left sides. However, the conditions of \ref{cond:matching-locs} ensure that these $4$ domains form $2$ pairs of complementary domains, which can bind together.  This also yields a polymer and configuration which is saturated and which also maximizes entropy and thus is stable, as is the configuration containing only polymers with correctly simulated computations.  Unfortunately, it also has the property of having output monomers which map to $M(k)$, and by our selection of $i$ and $k$, we know that $M(i) \ne M(k)$, and therefore the monomer collection $\vvc_i$ over $\mathcal{N}_i$ does not correctly simulate $M(i)$.

\section{Technical Details of Efficient simulation of Turing machines in GTBNs}\label{sec:GTBN-append}

In this section, we present the technical details of the proof of Theorem~\ref{thm:geo-TM} from Section~\ref{sec:GTBN}.

\begin{proof}

\begin{figure}
    \centering
    \includegraphics[width=\linewidth]{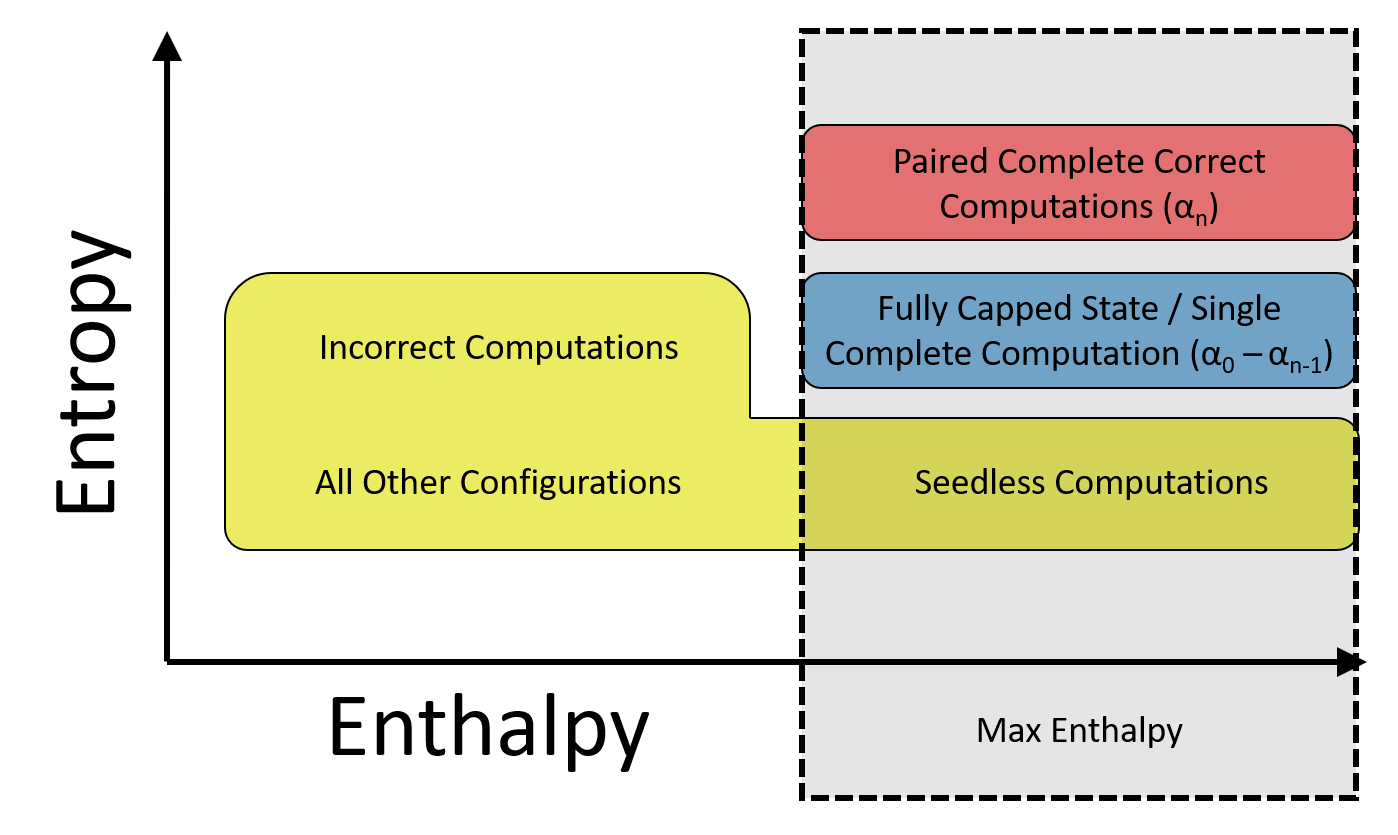}
    \caption{This graph shows gaps in both enthalpy and entropy of different polymers in our system. The most favorable polymer should always be the paired computations.}
    \label{fig:enthalpy_entropy_graph}
\end{figure}

We prove Theorem~\ref{thm:geo-TM} by construction.

Let $M = (Q, \Sigma, \Gamma, \delta, q_0, q_H)$ be the Turing machine which decides $L$, with state set $Q$, input alphabet $\Sigma$, tape alphabet $\Gamma$, transition function $\delta:(Q,\Gamma) \rightarrow (Q,\Gamma,D)$ (where $D$ is the set of directions $\{L,R\}$), start state $q_0$, and halting state $q_H$\footnote{For our construction, we assume that $M$ halts by entering $q_H$ with its tape head on a cell containing a $1$ for an accepting computation or a $0$ for a rejecting computation.}. We define the geometric monomers of $\calT_i$ based on the definition of $M$ and logically grouping them by functionality into a few main components. Note that these are logically very similar to those of the construction in Section~\ref{sec:hard-coded-comp}, but since they are geometric monomers we must also describe their shapes and the placements of their binding domains.  Also note that for easier visual depiction, the Turing machine simulations of this construction can be though of as oriented vertically, with each successive tape and machine state encoded by a row above the previous, while the other construction was oriented horizontally.

\begin{figure}[htp]
    \centering
    \includegraphics[width=\linewidth]{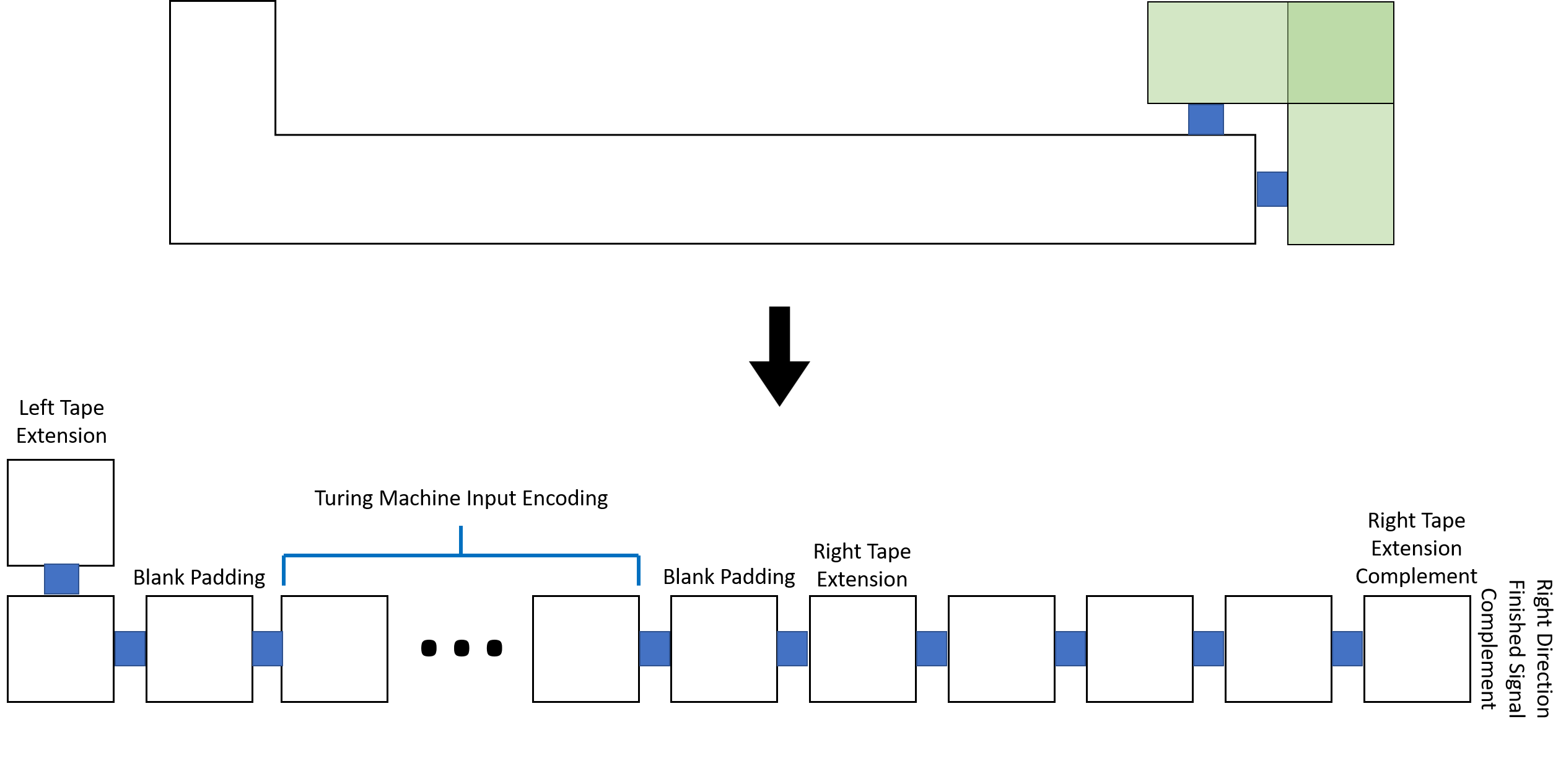}
    \caption{On the top is the illustration of the seed we use in other figures. The bottom diagram shows the specific domains that are present on the seed supertile. These include two tape extension domains, two padding domains, the TM input encoding, and the final pairing domains.}
    \label{fig:seed_breakdown}
\end{figure}

\subsection{Construction Components}

When defining specific geometric monomers, we use the convention $Monomer = \{ North, East, South, West \}$. Each direction is either a domain from $\mathcal{D} \cup \mathcal{D}^*$ or corresponds to a $\lambda$ to represent that no domain is present on that side. Note that, as we said earlier, geometric monomers are allowed to rotate, meaning this convention represents an ordering of the domains as opposed to which direction they are actually pointing. The geometric monomers of $\calT_i$ can be logically grouped into the following categories:

\begin{enumerate}

    \item The \emph{seed monomers}: We call this group of monomers $\calM_{seed}$. This group is infinite, $|\calM_{seed}| = \infty$, to account for the infinite unique strings that can be input into the simulated Turing machine. This group can further be broken down into two subgroups. The first being $\calM_{seed-input}$, an infinite subgroup that contains the monomers that actually encode the input and are the only monomers to change between different simulations of the same Turing machine. The other being $\calM_{seed-extra}$ which contains 9 monomers that are the same in every simulation and help the seed to operate correctly. The definitions of these monomers are given below and a visual representation is given in \ref{fig:seed_breakdown}.
    
    \begin{center}
    \begin{tabular}{c c c c c c}
        \emph{$M_{S1}$} $ = \{$ & $EXT\_LEFT^*,$ & $R^*,$ & $seed_{1,2},$ & $\lambda$ & $\}$ \\
        \emph{$M_{S2}$} $ = \{$ & $seed_{1,2}^*,$ & $seed_{2,3},$ & $\lambda,$ & $\lambda$ & $\}$ \\
        \emph{$M_{S3}$} $ = \{$ & $\_\_^*,$ & $seed_{3,I1},$ & $\lambda,$ & $seed_{2,3}^*$ & $\}$ \\
        \emph{$M_{I1}$} $ = \{$ & $[(I1),r]^*,$ & $seed_{I1,I2},$ & $\lambda,$ & $seed_{3,I1}^*$ & $\}$ \\
        \vdots & & & & \vdots \\
        \emph{$M_{I(k)}$} $ = \{$ & $[(Ik),r]^*,$ & $seed_{I(k),I(k+1)},$ & $\lambda,$ & $seed_{I(k-1),I(k)}^*$ & $\}$ \\
        \vdots & & & & \vdots \\
        \emph{$M_{I(n)}$} $ = \{$ & $[(In),r]^*,$ & $seed_{I(n),4},$ & $\lambda,$ & $seed_{I(n-1),I(n)}^*$ & $\}$ \\
        \emph{$M_{S4}$} $ = \{$ & $\_\_^*,$ & $seed_{4,5},$ & $\lambda,$ & $seed_{I(n),4}^*$ & $\}$ \\
        \emph{$M_{S5}$} $ = \{$ & $EXT\_RIGHT^*,$ & $seed_{5,6},$ & $\lambda,$ & $seed_{4,5}^*$ & $\}$ \\
        \emph{$M_{S6}$} $ = \{$ & $\lambda,$ & $seed_{6,7},$ & $\lambda,$ & $seed_{5,6}^*$ & $\}$ \\
        \emph{$M_{S7}$} $ = \{$ & $\lambda,$ & $seed_{7,8},$ & $\lambda,$ & $seed_{6,7}^*$ & $\}$ \\
        \emph{$M_{S8}$} $ = \{$ & $\lambda,$ & $seed_{8,9},$ & $\lambda,$ & $seed_{7,8}^*$ & $\}$ \\
        \emph{$M_{S9}$} $ = \{$ & $EXT\_RIGHT,$ & $done,$ & $\lambda,$ & $seed_{8,9}^*$ & $\}$
    \end{tabular}
    \end{center}
    
    Note that in the monomer descriptions, any character(s) in parentheses is a place holder. Therefore, $(I1)$, $(Ik)$, and $(In)$ represent the character that those monomers encode from ${0,1}$, whereas $(n)$ represents the length of the input string and $(k)$ is just the number of any character $1 < k < n$. Intuitively, the seed works by doing 4 things: first, it starts the left and right extension columns which will be further discussed in the next item. Second, it pads the input with two blank characters, which prevents the Turing machine from ever being able to run over the edge of the tape. Third, it "starts" the computation by providing the first $R^*$ domain. Lastly, it has a few monomers that attach to the right that will later connect two independent computations together to give a final entropy gap. Recall that $\vvc$ is defined using $\calM_i = \calM \cup \{m_i\}$ for GTBN $\mathcal{T}_i = (\mathcal{D},\mathcal{M}_i)$ where $M$ is all other monomers discussed other than $\calM_{seed}$ and $m_i$ is $\calM_{seed-extra}$ plus the monomers $M_{I(k)}$ for each $1 \le k \le |n|$.

    \item The \emph{computation monomers}: We call this set of monomers $\calM_{comp}$. It consists of two mutually exclusive subsets as well, the first being $\calM_{comp-passive}$ which consists of 6 monomers that pass values from line to line, each representing a value from the cross product of the line directions and tape alphabet, i.e. ${L,R} \times {\_\_, 0, 1}$. The definitions of the $\calM_{comp-passive}$ are below, substituting the value being passed into $(v)$:
    
    \begin{center}
    \begin{tabular}{c c c c c c}
        \emph{$M_{(v),L}$} $ = \{$ & $[(v),r]^*,$ & $L,$ & $[(v),l],$ & $L^*$ & $\}$ \\
        \emph{$M_{(v),R}$} $ = \{$ & $[(v),l]^*,$ & $R^*,$ & $[(v),r],$ & $R$ & $\}$
    \end{tabular}
    \end{center}
    
    The other subset is $\calM_{comp-transition}$ which consists of monomers that encode the logic of the transition function. The size of $\calM_{comp-transition}$ is asymptotically the same as the number of states times the size of the tape alphabet, i.e.  $|\calM_{comp-transition}| = O(Q \times \Gamma)$. The construction of these monomers is such that, for ever possible $(q,s) \in Q \times \gamma$, we generate the following monomers using the specifications in the transition function $\delta(q,s) = (q', s', D)$ and $\forall v \in \Gamma$.
    
    When $D = L$:
    
    \begin{center}
    \begin{tabular}{c c c c c c}
        \emph{$M_{(q,s)-skip}$} $ = \{$ & $[(q,s),l]^*,$ & $R^*,$ & $[(q,s),r],$ & $R$ & $\}$ \\
        \emph{$M_{(q,s)-move1}$} $ = \{$ & $[(s'),r]^*,$ & $L,$ & $[(q,s),l],$ & $(q')^*$ & $\}$ \\
        \emph{$M_{(q,s)-move2}$} $ = \{$ & $[(q',v),r]^*,$ & $(q'),$ & $[(v),l],$ & $L^*$ & $\}$
    \end{tabular}
    \end{center}
    
    When $D = R$:
    
    \begin{center}
    \begin{tabular}{c c c c c c}
        \emph{$M_{(q,s)-skip}$} $ = \{$ & $[(q,s),r]^*,$ & $L,$ & $[(q,s),l],$ & $L^*$ & $\}$ \\
        \emph{$M_{(q,s)-move1}$} $ = \{$ & $[(s'),l]^*,$ & $(q')^*,$ & $[(q,s),r],$ & $R$ & $\}$ \\
        \emph{$M_{(q,s)-move2}$} $ = \{$ & $[(q',v),l]^*,$ & $R^*,$ & $[(v),r],$ & $(q')$ & $\}$
    \end{tabular}
    \end{center}

    \item The \emph{tape extension monomers}: Because our goal is to simulate Turing machines that are not space-bounded, as opposed to the traditional TBN hard-coded system that simulates space-bounded Turing machines, we need additional monomers in our system that extend the tape of the Turing machine. We define this group of eight monomer types to be $\calM_{ext} = \{M_{L1}, M_{L2}, M_{L3}, M_{L4}, M_{R1}, M_{R2}, M_{R3}, M_{R4}\}$. Because of the small number of these monomers, we define each individually:
    
    \begin{center}
    \begin{tabular}{c c c c c c}
        \emph{$M_{L1}$} $ = \{$ & $extL_{4,1}^*,$ & $extL_{1,2},$ & $\lambda,$ & $\lambda$ & $\}$ \\
        \emph{$M_{L2}$} $ = \{$ & $extL_{2,3},$ & $L,$ & $EXT\_LEFT,$ & $extL{1,2}^*$ & $\}$ \\
        \emph{$M_{L3}$} $ = \{$ & $[\_\_,l]^*,$ & $R^*,$ & $extL_{2,3}^*,$ & $extL_{3,4}$ & $\}$ \\
        \emph{$M_{L4}$} $ = \{$ & $EXT\_LEFT^*,$ & $extL_{3,4}^*,$ & $extL_{4,1},$ & $\lambda$ & $\}$ \\
        \emph{$M_{R1}$} $ = \{$ & $extR_{4,1}^*,$ & $\lambda,$ & $\lambda,$ & $extR_{1,2}$ & $\}$ \\
        \emph{$M_{R2}$} $ = \{$ & $extR_{2,3},$ & $extR_{1,2}^*,$ & $EXT\_RIGHT,$ & $R$ & $\}$ \\
        \emph{$M_{R3}$} $ = \{$ & $[\_\_,r]^*,$ & $extR_{3,4},$ & $extR_{2,3}^*,$ & $L^*$ & $\}$ \\
        \emph{$M_{R4}$} $ = \{$ & $EXT\_RIGHT^*,$ & $\lambda,$ & $extR_{4,1},$ & $extR_{3,4}^*$ & $\}$
    \end{tabular}
    \end{center}
    
    These monomers don't encode values, just blank spaces that will pass up to higher rows to potentially be used in the Turing machine simulation. By utilizing these monomers, we ensure that the tape is always big enough for the computation to continue. An example of these supertiles is shown in figure \ref{fig:tape_extension_addition}.

    \item The \emph{end monomers}: We call this group of monomers $\calM_{end}$. These monomers only connect to the computation after the simulated TM goes into a halting state and act similar to the monomers in $\calM_{comp-passive}$ while also passing a signal that the computation is finished. The monomers will form one row if passing above the halting state $M_{(q_H),(v)}$ to the right or two rows if passing above the halting state to the left. This will ensure that the top row of the final computation monomer is a right-growing row, which is necessary for the eventual pairing of two complete computations. It is this final row that is input into the encoding function $E_{output}$ to read the result of the TM.
    
    For all halting states $(q_H)$ and output characters $(v) \in {0,1}$, should the halting state first be read with a left passing row, it attaches the following monomer:
    
    \begin{center}
    \begin{tabular}{c c c c c c}
        \emph{$M_{(v),haltL}$} $ = \{$ & $[(q_H,v),r]^*,$ & $L,$ & $[(q_H),(v),l],$ & $L^*$ & $\}$
    \end{tabular}
    \end{center}
    
    and on the subsequent right passing row, it will attach one of the following monomers, the same monomer that is checked for as input to the $E_{output}$ function:
    
    \begin{center}
    \begin{tabular}{c c c c c c}
        \emph{$M_{(v),haltR}$} $ = \{$ & $\lambda$ & $halt^*,$ & $[(q_H,v),r],$ & $R$ & $\}$
    \end{tabular}
    \end{center}
    
    and lastly, it will attach a string of the following monomers to finish off the top row, the last of which will also attach to the seed of the paired computation:
    
    \begin{center}
    \begin{tabular}{c c c c c c}
        \emph{$M_{(v),R1}$} $ = \{$ & $\lambda,$ & $halt^*,$ & $[(v),r],$ & $halt$ & $\}$ \\
    \end{tabular}
    \end{center}

    \item The \emph{capping monomers}: Caps in this system are very important because they prevent unwanted polymers from forming and make it easier to argue about the outcome of the system. Each of the monomer types talked about previously, the seed monomers, computation monomers, tape extension monomers, and end monomers all have caps. Each cap consists of three monomers, which form an "L" shape to reach both inputs of the capped monomer. The caps all have codomains $x^* \in \calD^*$, since every input in the system is from the set of domains $x \in \calD$, which prevents them from sticking together. While there are too many caps to explicitly define each individually, we do give an example of the seed cap:

    \begin{center}
    \begin{tabular}{c c c c c c}
        \emph{$M_{seed-cap1}$} $ = \{$ & $\lambda,$ & $seedcap_{1,2},$ & $EXT\_RIGHT^*,$ & $\lambda$ & $\}$ \\
        \emph{$M_{seed-cap2}$} $ = \{$ & $\lambda,$ & $\lambda,$ & $seedcap_{2,3},$ & $seedcap_{1,2}^*$ & $\}$ \\
        \emph{$M_{seed-cap3}$} $ = \{$ & $seedcap_{2,3}^*,$ & $\lambda,$ & $\lambda,$ & $halt^*$ & $\}$
    \end{tabular}
    \end{center}
    
    All other caps follow this same convention, just replacing $seedcap_{1,2}$ and $seedcap_{2,3}$ with other unique glues and $halt^*$ and $EXT\_RIGHT^*$ with the codomains that correspond to the inputs of the monomer to be capped.
    
\end{enumerate}

There are two key dynamics our system uses, the first of which is the formation of \emph{supertiles}. A supertile in our system is a set of monomers that have a unique set of matching domains/codomains between them but aren't capped with respect to one another. Therefore, it is always favorable to have them bind together, regardless of how other monomers in the system bind. The supertile itself can be capped, however, and all caps in this system are actually supertiles themselves as well. The other supertiles in our construction are the seed and tape extension pieces. Supertiles allows us to utilize the properties of polyominoes in our system while only requiring unit square shaped monomers. To simplify our arguments, we assume that the counts of all the monomers in the same supertile to be the same, as to avoid dealing with partially formed supertiles. Therefore, the counts of every monomer type from $\calM_{seed}$ in $\vvc$ are equal, the counts of every monomer type from $\calM_{ext}$ in $\vvc$ are equal, and the counts of every monomer type that makes up one cap are equal. Note that this assumption is analogous to having polyominoes replace supertiles in our system, which we avoided because we felt unit square geometric monomers were the TBN model that most naturally corresponded to the traditional aTAM model.

The other key dynamic is a \emph{fully capped state}. We utilize one other assumption in our system, that the count of each cap $x$ in the system is always greater than the number of monomers or supertiles that $x$ corresponds to. We use the term \emph{fully capped state} to refer to a configuration of the system in which every monomer or supertile that has a cap is bound to that cap through both domains. The individual monomers that bind supertiles are still bound in this state too, but no other polymers should be bound together. We use this state as a reference, somewhat of a neutral baseline, to show how correct computations are favored and all other non-capped polymers are unfavored.

In addition to the geometric monomers of $\calT$, we give \emph{relative counts} for the monomers in $\vvc$. This mean our system does not rely on exact counts of each monomer, but rather a group of assumptions about the relative counts. To show that every seed supertile is part of a complete paired computation polymer, we require that the number of cap supertiles in $\vvc$ is greater than the number of computation monomers, tape extension supertiles, and end monomers, whose counts in $\vvc$ are greater than the number of seed supertiles by a factor of $O(f(n)^2)$ (where $L \in DTIME(f(n))$ is the language being decided by the Turing machine $M$). In other words, we rely on the assumptions that: (1) every pair of seed supertiles has enough computation monomers, tape extension supertiles and end monomers to create a complete paired computation polymer (which is the reason for needing $O(f(n)^2)$ of those types relative to seed types), and (2) every computation monomer, tape extension supertile, and end monomer can have its own cap if it wasn't bound to anything else.

One last bit of terminology that we use in the following proofs is \emph{inputs} and \emph{outputs}. By inputs, we are referring to the two domains (not codomains) on the seed, computation monomers, tape extension supertiles, and end monomers. These two domains on each monomer/supertile mentioned are the domains that can bind to the associated cap. By outputs, we are referring to all of the codomains (not domains) on the seed, computation monomers, tape extension supertiles, end monomers and capping supertiles. While all computation monomers have two outputs, the seed can have any number, tape extension supertiles have three, and end monomers can have one or two. Notice that outputs only bind to inputs, and inputs only bind to outputs. The only domains present in our system that are not inputs or outputs are the unique domains that bind supertiles together.

\begin{claim}\label{clm:fcs_saturated}
    The fully capped state is saturated.
\end{claim}

\begin{proof}
    The argument for this claim is simple. Recall that a system is saturated if there are no labels in which both a domain and a codomain are left unbound. Any element $x \in \calD$ in our system falls into two categories, $\calD_{supertiles}$ which are unique domains used to form supertiles, and $\calD_{capped}$ which are input or output domains used to encode states, values, signals, etc. and always have a cap. In the fully capped state (and all other configurations), supertile monomers are always bound together forming the maximum number of supertiles and ensuring as many domains in $\calD_{supertiles}$ are bound as possible. For domains in $\calD_{capped}$, because the number of complete cap supertiles for each pair of inputs $(I1, I2)$ in $\vvc$ is greater than the number of monomers or supertiles of the type that corresponds to those inputs, every domain in $\calD_{capped}$ will also be in a bond. Therefore, no domain in the system will be left unbound, meaning that the fully capped state of $\vvc$ is saturated.
\end{proof}

\subsection{The Desired Configuration}

We now argue about the desired configuration of $\vvc_i$, which we call $\alpha$, that consists of 5 types of polymers. 1) polymers containing two instances of the seed supertile, computation monomers corresponding to two correct and complete computations of $M$ on input $i$, a series of left tape extension supertiles and a series of right tape extension supertiles, and one final partial row of end monomers that, along with the final right tape extension supertiles, connect to one instance of the seed 2) polymers that consist of a single capping supertile bound to a single computation monomer and bound by $2$ binding domains, 3) polymers that consist of a single capping supertile bound to a single tape extension supertile and bound by $2$ binding domains, 4) polymers that consist of a single capping supertile bound to a single end monomer and bound by $2$ binding domains, and 5) polymers that are single capping supertiles. For polymers described by (1), we let $n$ denote the number of computation monomers, tape extension supertiles, and end monomers in such a polymer plus one. We will prove that $\alpha$ is saturated, and then show that $\alpha$ is the only stable configuration of $\vvc_i$.

\begin{claim}\label{clm:geo_saturated}
    $\alpha$ is saturated.
\end{claim}

\begin{proof}

    Similar to the proof of claim \ref{clm:saturated}, we define a sequence of configurations $\alpha_i \in seq(\alpha)$ for $0 \le i \le n$ where $\alpha_0$ is the fully capped state, $\alpha_n = \alpha$, and for every $i > 1$, $\alpha_i$ is saturated if $\alpha_{i-1}$ is saturated. Furthermore, we also define $p(\alpha_i)$ to be the polymer that contains the seed(s) in the configuration $\alpha_i$ and $p(seq(\alpha))$ to be the set of all polymers $p(\alpha_i)$ for $0 \le i \le n$, i.e. a seed supertile, the complete paired computation, and everything in between. The difference between every $\alpha_{i-1}$ and $\alpha_i$ is the transition of just one computation monomer, tape extension supertile, or end monomer from a polymer with its cap to the polymer $p(\alpha_i)$. Notice at any point in the sequence before $\alpha_n$ there will always be one (and only one) concave corner at the end of the partially formed top row of $p(\alpha_i)$, where two unbound output domains will be adjacent to the same unit square location. Because of the deterministic nature of $M$, there is exactly one type of computation monomer, tape extension supertile, or end monomer (or seed when $i = n-1$) with input domains that correspond to the unbounded output domains at this location. The piece that corresponds to this input is the piece that transitions from its cap to $p(\alpha_i)$ in the step $\alpha_{i-1} \rightarrow \alpha_i$. The only exception is the final step $\alpha_{n-1} \rightarrow \alpha_n$ which is the pairing of two complete computation polymers together. This step works by starting with two independent computations polymers that have completed their respective instances of the simulation and have one final concave corner location with the unbound output domains $halt^*$ and $EXT\_RIGHT^*$ exposed. These correspond to the capped input domains on the far right of the seed supertile. In this step $\alpha_{n-1} \rightarrow \alpha_n$, both seeds break off their respective cap and bind to the unbound output domains on the other computation polymer.
    
    We view the enthalpy and entropy difference of each step $\alpha_{i-1} \rightarrow \alpha_i$ as independent of the rest of the system, considering the rest of the system is constant while the transition is happening. From figure \ref{fig:monomer_additions}, it is easy to see how the transition of computation monomers, tape extension pieces, and end monomers from their respective caps to $p(\alpha_i)$ is always enthalpy and entropy neutral. Each transition breaks two bonds and makes two more bonds and starts with two polymers and ends with two polymers. Fast forwarding to the end of the sequence, it is also clear from figure \ref{fig:complete_computation_gaps} that the overall enthalpy and entropy of the $\alpha_{n-1}$ configuration is the same as the overall enthalpy and entropy of $\alpha_0$ (in this example, 86 bonds and 43 polymers). The last step of the sequence, illustrated in figure \ref{fig:paired_computation_gap}, shows how the pairing of two complete computation polymers is enthalpy neutral (breaking 4 bonds and making 4 bonds) but actually gives an entropy bonus of $+1$ to the system (starting with 2 polymers and ending with 3 polymers). Because this final configuration has the same enthalpy as the fully capped state, which is saturated, we know this configuration is saturated as well.

\end{proof}

\begin{claim}\label{clm:one_output}
    A polymer in $\alpha_n$ that contains two seeds will contain two monomers of the type $M_{(v),haltR}$ that encode the same value of $v$.
\end{claim}

\begin{proof}
    First, notice the south input to the $M_{(v),haltR}$ monomer is the $[(q_H,v),r]$ input domain. The corresponding output domain $[(q_H,v),r]^*$ is only present on the north face of the $M_{(q,s)-move2} \in \calM_{comp-transition}$ monomer. However, the $\calM_{comp-transition}$ set is designed so that every monomer in the set can only bind one of its input domains to one of the output domains of another monomer in the set or the seed. Therefore, the seed starts one path of monomers exclusively in $\calM_{comp-transition}$. This path ends when either $M_{(v),haltL}$ or $M_{(v),haltR}$ is attached. If $M_{(v),haltL}$ is attached first, it subsequently attaches one monomer of the $M_{(v),haltR}$ type. Because this path is started by the seed and the main polymer of $\alpha$ has two seeds, this polymer will also contain two monomers of the type $M_{(v),haltR}$.
    
    Because the simulated Turing machine is deterministic, every step from $\alpha_i$ to $\alpha_{i+1}$ has a unique monomer or supertile making the transition from its cap to the computation polymer. Therefore, the two independent computations that are paired together in the main polymer of $\alpha$ must be identical, therefore ensuring that the two monomers of type $M_{(v),haltR}$ encode the same value $v$.
\end{proof}

\begin{figure}[htp]
\centering
    \begin{subfigure}{1\linewidth}
        \includegraphics[width=1\linewidth]{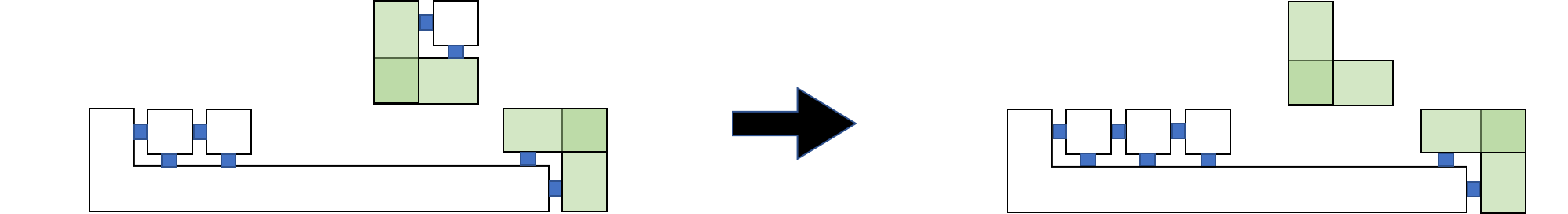}
        \caption{Computation Monomer (End Monomers)}
        \label{fig:computation_monomer_addition}%
        \end{subfigure}
        \begin{subfigure}{1\linewidth}

        \includegraphics[width=1\linewidth]{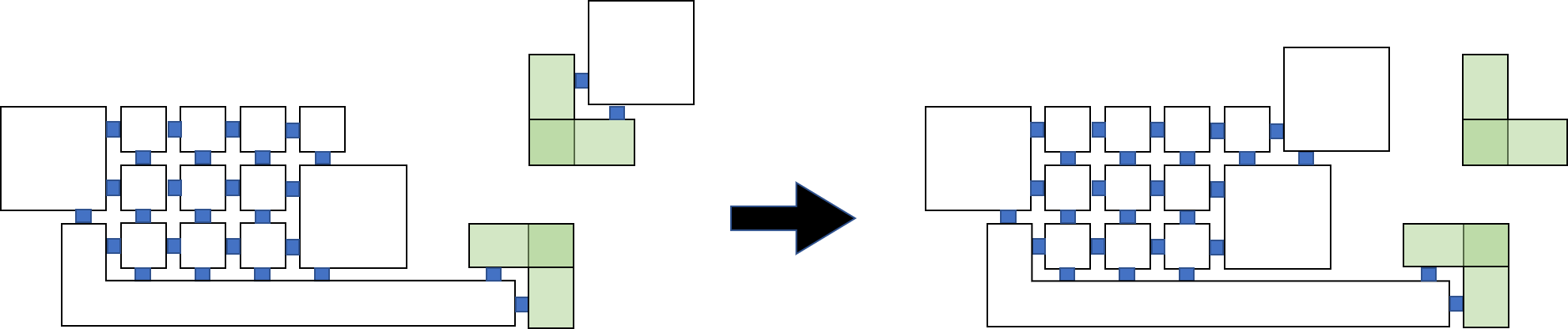}
        \caption{Tape Extension Piece}
        \label{fig:tape_extension_addition}%
        \end{subfigure}
  \caption{The addition of all correct pieces to the seed are enthalpy and entropy neutral. Each breaks two bonds with the cap and forms two more with the computation while also remaining at two polymers before and after the transition.}
  \label{fig:monomer_additions}
\end{figure}

\begin{figure}
    \centering
    \includegraphics[width=\linewidth]{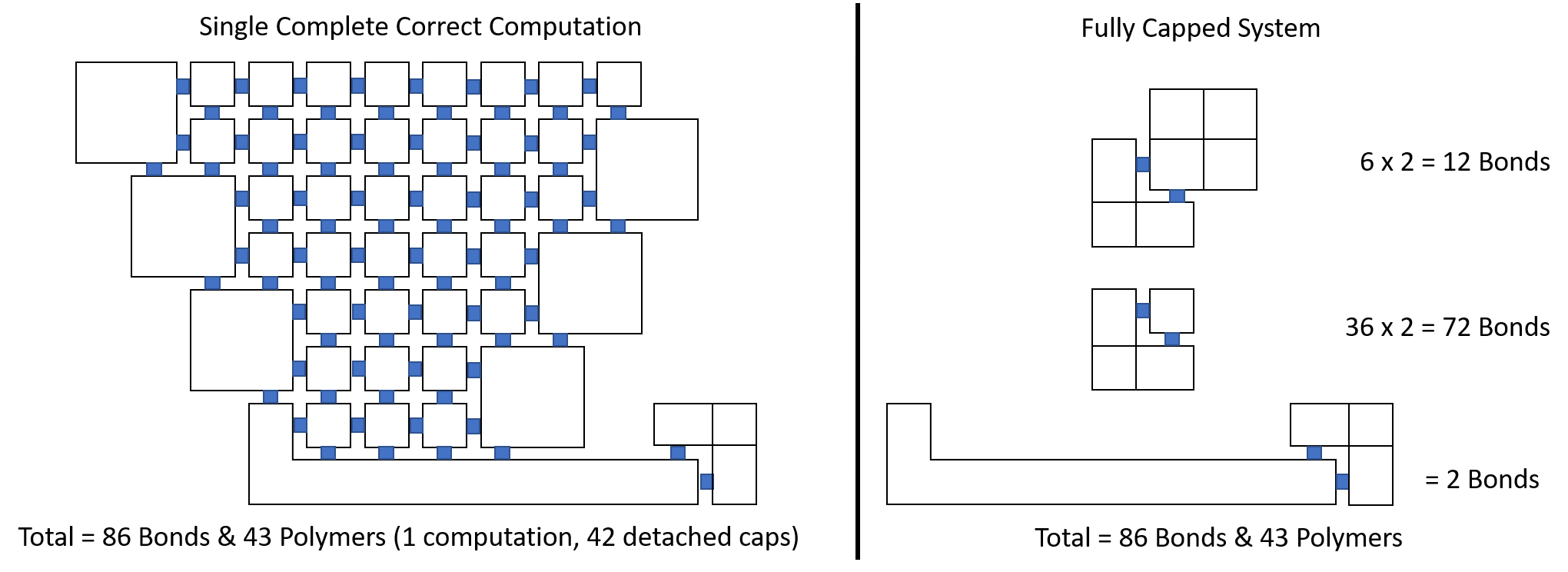}
    \caption{On the left is a completed computation that has not yet been paired. Compared to the fully capped state, it has the same number of bonds and polymers, meaning it still has neutral enthalpy and entropy.}
    \label{fig:complete_computation_gaps}
\end{figure}

\subsection{Incorrect Computations}

We already showed that $\alpha$ is saturated. In order to show that $\alpha$ is stable, we need to show that all other saturated configurations in $\vvc$ have less entropy. We start by showing the majority of configurations in $\vvc$ aren't saturated at all. Then we show the few that are also saturated create less polymers than $\alpha$.

First, we prove a couple claims that help us show that configurations not in $seq(\alpha)$ are unstable.

\begin{claim}\label{clm:unbound_input}
    Assuming an excess of caps, any configuration of $\vvc$ that contains at least one monomer/supertile with an unbound input is not saturated.
\end{claim}

\begin{proof}
    Given an excess of caps, the fully capped state is possible. Since the fully capped state has every input of all monomers and supertiles in the $\vvc$ bound, and since inputs are the limiting factor on the maximum number of bonds in $\vvc$, then any configuration that does not have every input bound has not formed the maximum number of bonds and is therefore unsaturated.
\end{proof}

\begin{claim}\label{clm:two_caps}
    Assuming an excess of caps, any configuration $\vvc$ that includes a polymer $p$ consisting of 2 or more caps is not stable.
\end{claim}

\begin{proof}
    This only requires a simple counting proof. Assume $p$ consists of $M$ monomers and supertiles (each supertile counts as one) and $N$ caps. Considering the caps of the $M$ monomers and supertiles, this accounts for $M + 1$ polymers. However, in a fully capped state, $p$ could be broken up with its caps into $M + N$ polymers. Therefore, whenever $N > 1$, the number of polymers in the fully capped state, $M + N$, is greater the number of polymers in $\vvc$, $M + 1$, meaning $\vvc$ is not stable.
\end{proof}

\begin{claim}\label{clm:no_rotation}
    If any two monomers $m1$ and $m2$ in a single polymer are connected through a path of other monomers and supertiles that does not include the seed, $m1$ and $m2$ must have the same rotation.
\end{claim}

\begin{proof}
    Although rotation of monomers and polymers relative to each other is permitted in the GTBN, we design domains so that they are only complementary in north/south or east/west pairs, except for the seed/end monomer connection and the seed/right tape extension supertile connection that allow paired computations to combine. Therefore, if two monomers are bound together in a polymer of $\vvc$ by a path that does not contain the seed, one cannot be rotated relative to the other.
\end{proof}

We are also going to make use of the idea of a \emph{subpolymer}. We define a subpolymer of a polymer $p$ to be another polymer made up of a set of ordered pairs of monomer types and corresponding locations that is a subset of the set of ordered pairs in $p$, i.e. a it is the same polymer as $p$ but with some monomers potentially missing.  Additional terminology that we will use in these proofs is \emph{signal}, \emph{pumping}, and \emph{stopper}. These proofs are all focused on finding an unbound input in a polymer. By signal, we are referring to any domain that has an input on one face of a monomer and the corresponding output on the opposite face. Examples of signals are $L$, $R$, $EXT\_LEFT$, $EXT\_RIGHT$, and $halt$. Pumping is the idea that trying to bind an unbound input instance of a signal with one of these monomers will just expose another input instance. Finally, a stopper is a monomer or supertile that can bind to an unbound input instance of a signal without exposing another input instance, thereby stopping the signal from pumping. For example, the left tape extension supertile acts as a stopper for the $R$ signal because it has an $R^*$ output domain but does not have an $R$ input domain. Likewise, the seed is a stopper for the $EXT\_LEFT$ signal because it has an $EXT\_LEFT^*$ output domain but does not have an $EXT\_LEFT$ input domain.

\begin{claim}
    Any configuration of $\vvc$ that includes a polymer $p_{incorrect}$ that contains the seed but is not in the set $p(seq(\alpha))$ is unsaturated.
\end{claim}

\begin{proof}

    This proof utilizes the fact that our TM is deterministic, and therefore, every pair of inputs $(I1, I2)$ where $I1,I2 \in \calD$ corresponds to one unique monomer. Take that polymer $p_{incorrect}$ that contains the seed but differs from any polymer in $p(seq(\alpha))$. Let $p_{sub}$ be the largest polymer such that $p_{sub} \in p(seq(\alpha))$ and $p_{sub}$ is a subpolymer of $p_{incorrect}$. Notice that such a polymer must exist, since $p(\alpha_0)$ is the seed and we are only considering polymers that contain the seed. Also notice that all polymers in $p(seq(\alpha))$ will only ever have 5 unique types of output codomains (and 0 input domains) unbound and exposed: $EXT\_LEFT^*$, $[(v),\{l,r\}]^*$, $[(q',v'),\{l,r\}]^*$, $EXT\_RIGHT^*$, and one of the three signal output domains $\{L^*,R^*,halt^*\}$. Since $p_{incorrect}$ is different than any polymer in $p(seq(\alpha))$, then at least one of these output domains extending from $p_{sub}$ must be bound to a monomer or supertile that is different than the monomer or supertile that it is bound to in $\alpha$.
    
    In the first case, assume the incorrect extension was bound to $p_{sub}$ by the output domains of either $EXT\_LEFT^*$ or $EXT\_RIGHT^*$. Since the corresponding input domains $EXT\_LEFT$ and $EXT\_RIGHT$ are unique to the tape extension supertiles, these output domains cannot attach the incorrect piece. The only other mistake that can happen is if a tape extension supertile attaches past the final row of computation. Assuming this had happened with a right supertile, this would expose an $R$ input domain. This input could either be unbound, causing the configuration to be unsaturated, or it could be attached to another seed, a computation monomer, or a left tape extension supertile, the three pieces with $R^*$ output domain. It's easy to see another seed or a left tape extension supertile would be blocked geometrically from attaching here. A computation monomer could be attached, however, it would have to pump the $R$ signal to avoid an unbound input. There are two subcases that could happen. One, it would reach the edge of a partially formed top row of $p_{sub}$. Here, it could be attached to a computation monomer that was hanging off the edge, but this computation monomer would have an input to the south that was unbound. If the unbound input was also bound by a computation monomer, that computation monomer would be the correct piece, thereby contradicting that $p_{sub}$ is maximal. The hanging monomer could also have a cap, which would cause $p_{incorrect}$ to have another cap in addition to the seed cap, making it unstable by claim \ref{clm:two_caps}. Instead of additional computation monomers, the last computation monomer of the pumping chain could also be attached to a left tape extension supertile. However, since the inputs of the supertile would be pressed against $p_{sub}$, the inputs would have to be connected to outputs of $p_{sub}$. If this was the case, the supertile would be correct, contradicting that $p_{sub}$ is maximal. Two, it would pump the signal to a location that would require a monomer that has the input domain $[(q_H,v),r]$ on its south face and output domain $R^*$ on its east face, of which there are none. Assuming the extra tape extension supertile was on the left, the supertile would require an $L$ input domain to be bound. This would yield two symmetrical subcases. One, the signal would pump with computation monomers to the edge of a partially formed top row of $p_{sub}$, which we just showed leads to the configuration being unstable. The signal could also pump until it reached the location north of $M_{(v),haltR}$ monomer. Because the $M_{(v),haltR}$ monomer has no output on the north face and all computation pieces have an input on the south face, no monomer could be in this location and have its south input bound, meaning there must be an unbound input somewhere, making the configuration unsaturated.
    
    In the second case, $[(v),\{l,r\}]^*$ and $[(q',v'),\{l,r\}]^*$ can only attach a computation monomer or end monomer. Because these domains encode the direction of the next row, any attached computation or end monomer must also have an input domain in the opposite direction. In other words, any monomer attached to $[(v),l]^*$ or $[(q', v'),l]^*$ must have an input domain $L$ on the east face, likewise with $[(v),r]^*$ or $[(q',v'),r]^*$ and an input domain $R$ (or $halt$) on the west face. In order to have this input also bound, there must be another computation monomer or end monomer connected to this input, with another $L$ or $R$ (or $halt$) on its east or west face, respectively. This signal would have to pump to avoid an unbound input. There's two possibilities for the end of this chain: 1) it eventually connects to the signal output domain $\{L^*,R^*,halt^*\}$ on $p_{sub}$. If this happens, the last monomer would have two input domains that match two output domains from $p_{sub}$, meaning it was the correct tile, contradicting that $p_{sub}$ was maximal. 2) the chain of monomers would pump to the edge of the partially formed top row of $p_{sub}$. Similar again to the argument in case one, the input of the last monomer in the chain could either be attached to: (a) another computation monomer hanging over the edge or (b) a tape extension supertile. For (a), the hanging computation monomer would require another computation monomer to the south, a bound cap, or would have an unbound input domain. Another computation monomer to the south would have to match two input domains to $p_{sub}$, contradicting that $p_{sub}$ was maximal. A bound cap would be the second on $p_{incorrect}$, in addition to the seed cap, making the configuration unstable by claim \ref{clm:two_caps}. For (b), the supertile would have its inputs pressed against $p_{sub}$, forcing it to bind to $p_{sub}$ to avoid an unbound input. This, however, would make the supertile the correct piece for that location, contradicting that $p_{sub}$ was maximal.
    
    In the last case, if any incorrect monomers or supertiles are bound to the signal output domain $\{L^*,R^*,halt^*\}$, that monomer or supertile must also have an input domain to the south that binds with $p_{sub}$. However, if this monomer or supertile matches both of its input domains to two output domains of $p_{sub}$, then it must correspond to the correct monomer or supertile for that location, once again contradicting that $p_{sub}$ was maximal. These three cases show that any erroneous addition to a polymer in $p(seq(\alpha))$ make the configuration containing that polymer unsaturated, and thus unstable.

\end{proof}

\subsection{Seedless Polymers}

Now we will show that configurations in $\vvc$ that include polymers that do not contain the seed and do contain more than one element from the set of all computation monomers, tape extension supertiles, and end monomers are unstable. We start by proving a few claims about signals.

\begin{claim}\label{clm:input_ext_left}
    If a polymer does not contain the seed and does contain a monomer with the input domain $EXT\_LEFT$, it has either an exposed instance of the $EXT\_LEFT$ domain or an attached cap.
\end{claim}

\begin{proof}
    First, notice that the only monomer or supertile that contains the $EXT\_LEFT$ domain as an output is the left tape extension supertile. However, this supertile also has the domain as an input. Therefore, if a polymer contains the $EXT\_LEFT$ input domain, it can either be exposed, capped, or bound to an instance of the tape extension supertile. The last case is the only implication not included in the claim, but it exposes another $EXT\_LEFT$ input domain, causing the signal to pump. Whereas, in polymers with the seed, the seed acts as the stopper for this signal, but without the seed, no other monomers or supertile can act as a stopper. Because polymers cannot be infinite, eventually one of these input domains will have to be left exposed or capped. Note that this proof is analogous for $EXT\_RIGHT$.
\end{proof}
    
\begin{claim}\label{clm:input_R}
    If a polymer does not contain the seed and does contain a monomer with input domain $R$, it has either an exposed instance of the $R$ domain, an exposed instance of the $EXT\_LEFT$ domain, or an attached cap.
\end{claim}

\begin{proof}
    Using the same reasoning as the proof for claim \ref{clm:input_ext_left}, notice that the only monomers and supertiles that contain the $R$ domain as an output also have the domain as an input. The only exception this time is the left tape extension supertile, which does have $R^*$ as an output and does not have $R$ as an input. Therefore, if a polymer contains the $R$ input domain, it can either be exposed, capped, bound to a left tape extension supertile, or bound to another monomer that also has $R$ as an input domain. The first two cases are implications in the claim, whereas the third case, being bound to a left tape extension supertile, exposes a $EXT\_LEFT$ input domain, which also leads to the implications of claim \ref{clm:input_R} using claim \ref{clm:input_ext_left}. The final case, being bound to another monomer that also has $R$ as an input domain, creates a new instance of the $R$ input domain causing the signal to pump. Because polymers cannot be infinite, eventually one of the first three cases will happen, eventually leading to an exposed input domain or a cap. Note that this proof shows $R \rightarrow EXT\_LEFT$ and is analogous for $halt \rightarrow R$ and for $L \rightarrow EXT\_RIGHT$.
\end{proof}

\begin{claim}
    Any configuration of $\vvc$ that includes a polymer $p_{seedless}$ that does not contain the seed but does contain two or more elements from the set off all computation monomers, tape extension supertiles, and end monomers is unstable.
\end{claim}

\begin{proof}
    
    We're going to break up this proof into three cases. In the first case, $p_{seedless}$ has a tape extension supertile. Without loss of generality, assume it is a left tape extension supertile. This supertile has two inputs, $EXT\_LEFT$ and $L$. By claim \ref{clm:input_ext_left}, $EXT\_LEFT$ implies there is an unbound input domain or an attached cap, and by claim \ref{clm:input_R}, $L$ implies there is another unbound input domain or attached cap that is independent of the first. Since $p_{seedless}$ must have at least one unbound input or two attached caps, $p_{seedless}$ is unstable.
    
    In the second case, $p_{seedless}$ has one $L$ input domain and one $R$ (or $halt$) input domain, or two of the same input domain on different rows, i.e. connected through north-south domains. If the signals are different, claim \ref{clm:input_R} can apply to both, indicating that $p_{seedless}$ has at least one exposed input or two attached caps. If they are the same signal on different rows, there is the possibility that they could turn into the same $EXT\_LEFT$ or $EXT\_RIGHT$ signal. However, if they did, $p_{seedless}$ must contain at least one tape extension supertile, making this case one again which we already showed is unstable.
        
    The only other case is when $p_{seedless}$ is a single row of computation (or end) monomers, i.e. a polymer with just one $L$ or $R$ (or $halt$) signal. However, since these polymers consist of computation (or end) monomers, each of which has a second input on its south face, the overall polymer will still have input domains, in the form of values from another row, that need to be addressed. If these input domains are left unbound, then the configuration isn't maximizing bonds and is therefore unsaturated. If these domains are bound, it must be with additional computation monomers, which will create another instance of an $L$ or $R$ in another row, giving the polymer two signals which we already showed is unstable. One important note is that, if the original $L$ or $R$ (or $halt$) signal is capped, it will also cap one of these value input domains. However, since we are arguing about polymers with at least two non-cap monomers or supertiles, there must be at least one other value input domain, which can only be bound by the previously mentioned additional row of computation monomers. Examples of unstable spurious computations are shown in \ref{fig:seedless_computation_gaps}.

\end{proof}

\begin{figure}
    \centering
    \includegraphics[width=\linewidth]{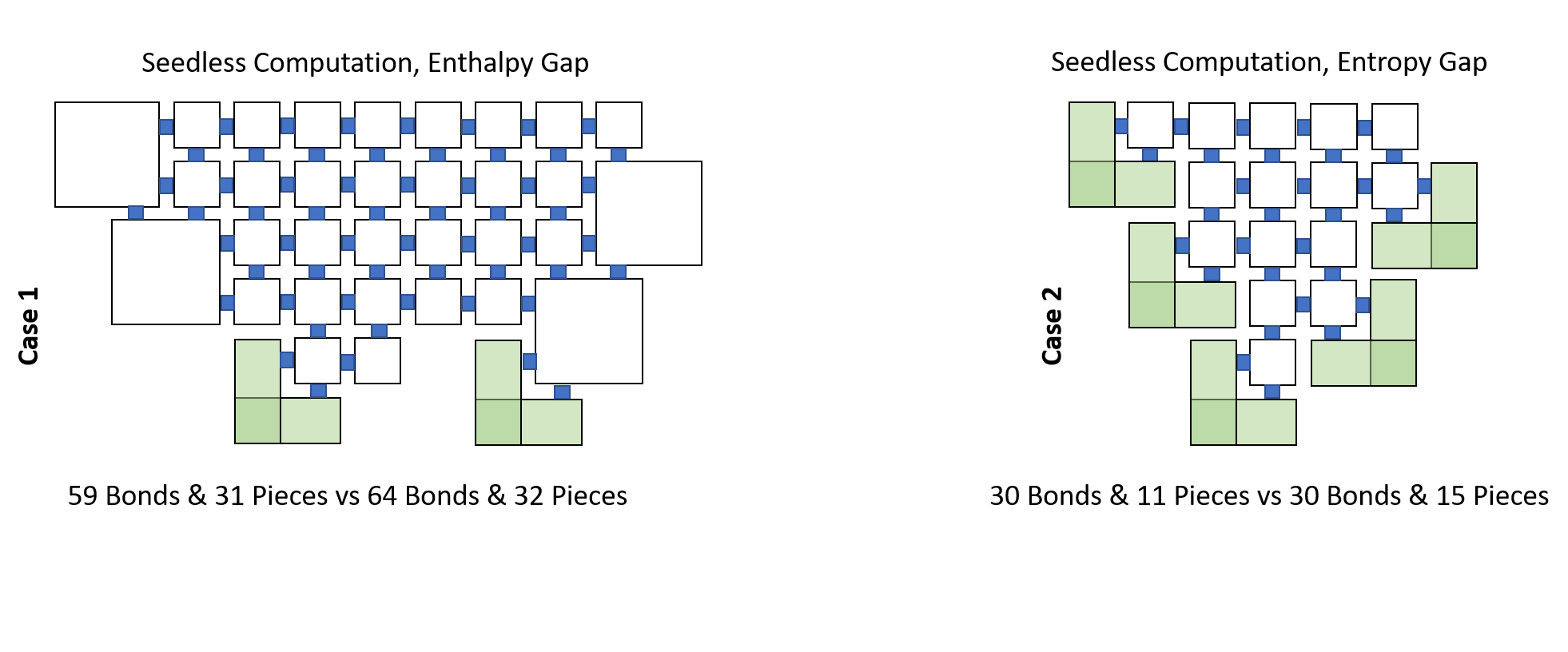}
    \caption{On the left is an example of a spurious computation with exposed domains. Because these domains are unbounded, there is an enthalpy gap between this polymer with its detached caps and a fully capped state. On the right is a special case where all the exposed input domains are still bound to their caps. In this case, there is an entropy gap with a fully capped state. To clarify, these two subsystems are not being compared to each other, but the statistics below each image are the comparisons between each polymer and a fully capped state.}
    \label{fig:seedless_computation_gaps}
\end{figure}

To summarize the previous section, we have shown that all configurations in $seq(\alpha)$ are saturated, from the $\alpha_0$, the fully capped state, to $\alpha$, the configuration in which every seed is bound into a complete paired computation. We also showed that the final configuration $\alpha$ has a higher entropy than every other configuration in $seq(\alpha)$. Then we showed that any configuration with a polymer that contains the seed but is not an element of $p(seq(\alpha))$ has unbound input domains and is therefore not saturated. Finally, we showed that any polymer that does not contain the seed and does contain at least two elements from the set of all computation monomers, tape extension supertiles, end monomers, and caps is either unsaturated or has a lower entropy than $\alpha_0$, the fully capped state. The accumulation of all these results proves that $\alpha$ is the single stable configuration of $\vvc$.

Now that we have shown how $\vvc$ simulates $M$, that $p(\alpha)$ will always contain two matching monomers of type $M_{(v),haltR}$, and proved that $\alpha$ is the stable configuration of $\vvc$, all that is left to do is to show how $E_{output}$ works. $E_{output}$ works by looking at the stable configuration of $\vvc$, analyzing the definition of any polymer that contains the seed supertile ($p(\alpha)$), and outputs a $0$ if that polymer contains two instances of $M_{0,haltR}$ or $1$ if that polymer contains two instance of $M_{1,haltR}$. The proof that two monomers of one of these two types will be in every instance of the final complete paired computation polymer is a combination of the proof that $p(\alpha)$ contains these two monomers and the proof that $\alpha$ is the stable configuration. This final output $0$ or $1$ corresponds to the output of the Turing machine $M$.

\end{proof}

\end{document}